\documentclass[11pt]{article}

\usepackage{amsthm}
\usepackage{graphicx} 
\usepackage{array} 

\usepackage{amsmath, amssymb, amsfonts, verbatim}
\usepackage{hyphenat,epsfig,subcaption,multirow}
\usepackage[font=small,labelfont=bf]{caption}
\usepackage{nicefrac}
\usepackage{paralist}

\usepackage{tocloft}

\usepackage[usenames,dvipsnames]{xcolor}
\usepackage[ruled]{algorithm2e}

\DeclareFontFamily{U}{mathx}{\hyphenchar\font45}
\DeclareFontShape{U}{mathx}{m}{n}{
      <5> <6> <7> <8> <9> <10>
      <10.95> <12> <14.4> <17.28> <20.74> <24.88>
      mathx10
      }{}
\DeclareSymbolFont{mathx}{U}{mathx}{m}{n}
\DeclareMathSymbol{\bigtimes}{1}{mathx}{"91}

\usepackage{tcolorbox}
\tcbuselibrary{skins,breakable}
\tcbset{enhanced jigsaw}

\usepackage[normalem]{ulem}
\usepackage[compact]{titlesec}

\definecolor{DarkRed}{rgb}{0.5,0.1,0.1}
\definecolor{DarkBlue}{rgb}{0.1,0.1,0.5}

\usepackage{nameref}
\definecolor{ForestGreen}{rgb}{0.1333,0.5451,0.1333}
\definecolor{Red}{rgb}{0.9,0,0}
\usepackage[linktocpage=true,
	pagebackref=true,colorlinks,
	linkcolor=DarkRed,citecolor=ForestGreen,
	bookmarks,bookmarksopen,bookmarksnumbered]
	{hyperref}
\usepackage[noabbrev,nameinlink]{cleveref}
\crefname{property}{property}{Property}
\creflabelformat{property}{(#1)#2#3}
\crefname{equation}{eq}{Eq}
\creflabelformat{equation}{(#1)#2#3}

\usepackage{bm}
\usepackage{url}
\usepackage{xspace}
\usepackage[mathscr]{euscript}

\usepackage{tikz}
\usetikzlibrary{arrows}
\usetikzlibrary{arrows.meta}
\usetikzlibrary{shapes}
\usetikzlibrary{backgrounds}
\usetikzlibrary{positioning}
\usetikzlibrary{decorations.markings}
\usetikzlibrary{decorations.pathreplacing} 
\usetikzlibrary{patterns}
\usetikzlibrary{calc}
\usetikzlibrary{fit}
\usetikzlibrary{decorations}

\usepackage[framemethod=TikZ]{mdframed}

\usepackage[noend]{algpseudocode}
\makeatletter
\def\BState{\State\hskip-\ALG@thistlm}
\makeatother

\usepackage{cite}
\usepackage{enumitem}
\setlist[itemize]{leftmargin=20pt}
\setlist[enumerate]{leftmargin=20pt}

\usepackage[margin=1in]{geometry}

\usepackage{thmtools}
\usepackage{thm-restate}

\newtheorem{theorem}{Theorem}
\newtheorem{lemma}{Lemma}[section]
\newtheorem{proposition}[lemma]{Proposition}

\newtheorem{invariant}[lemma]{Invariant}

\crefname{question}{Question}{Questions}

\newtheorem{problem}{Problem}

\newtheorem*{claim*}{Claim}
\newtheorem*{assumption*}{Assumption}
\newtheorem*{proposition*}{Proposition}
\newtheorem*{lemma*}{Lemma}

\newtheorem{observation}[lemma]{Observation}

\newtheorem*{theorem*}{Theorem}

\crefname{lemma}{Lemma}{Lemmas}
\crefname{claim}{claim}{claims}
\crefname{property}{Property}{Properties}
\crefname{invariant}{Invariant}{Invariants}

\newtheorem{mdresult}{Result}
\newenvironment{result}{\begin{mdframed}[backgroundcolor=lightgray!40,topline=false,rightline=false,leftline=false,bottomline=false,innertopmargin=5pt]\begin{mdresult}}{\end{mdresult}\end{mdframed}}

\theoremstyle{definition}

\newtheorem*{mdproblem*}{Problem}
\newenvironment{Problem*}{\begin{mdframed}[hidealllines=false,innerleftmargin=10pt,backgroundcolor=gray!10,innertopmargin=5pt,innerbottommargin=5pt,roundcorner=10pt]\begin{mdproblem*}}{\end{mdproblem*}\end{mdframed}}
\newtheorem{mddefinition}[lemma]{Definition}

\newtheorem*{mddefinition*}{Definition}
\newenvironment{Definition*}{\begin{mdframed}[hidealllines=false,innerleftmargin=10pt,backgroundcolor=white!10,innertopmargin=5pt,innerbottommargin=5pt,roundcorner=10pt]\begin{mddefinition*}}{\end{mddefinition*}\end{mdframed}}
\newtheorem{mdremark}{Remark}

\newenvironment{ourbox}{\begin{mdframed}[hidealllines=false,innerleftmargin=10pt,backgroundcolor=white!10,innertopmargin=5pt,innerbottommargin=5pt,roundcorner=10pt]}{\end{mdframed}}

\newtheorem{mdalgorithm}{Algorithm}
\newenvironment{Algorithm}{\begin{ourbox}\begin{mdalgorithm}}{\end{mdalgorithm}\end{ourbox}}

\newtheorem{mddistribution}{Distribution}

\newtheorem{mddistributionpart}{Part}

\allowdisplaybreaks

\renewcommand{\qed}{\nobreak \ifvmode \relax \else
      \ifdim\lastskip<1.5em \hskip-\lastskip
      \hskip1.5em plus0em minus0.5em \fi \nobreak
      \vrule height0.75em width0.5em depth0.25em\fi}

\renewcommand{\leq}{\leqslant}
\renewcommand{\geq}{\geqslant}

\newcommand{\Ot}{\ensuremath{\widetilde{O}}}
\newcommand{\eps}{\ensuremath{\varepsilon}}

\newcommand{\bracket}[1]{\left[#1\right]}
\newcommand{\paren}[1]{\ensuremath{\left(#1\right)}\xspace}
\newcommand{\card}[1]{\left\vert{#1}\right\vert}

\newcommand{\abs}[1]{\ensuremath{|#1|}}

\newcommand{\prob}[1]{\Pr\paren{#1}}
\newcommand{\expect}[1]{\Exp\bracket{#1}}

\newcommand{\set}[1]{\ensuremath{\left\{ #1 \right\}}}
\newcommand{\poly}{\mbox{\rm poly}}

\DeclareMathOperator*{\Exp}{\ensuremath{{\mathbb{E}}}}
\DeclareMathOperator*{\Prob}{\ensuremath{\textnormal{Pr}}}
\renewcommand{\Pr}{\Prob}

\newcommand{\II}{\ensuremath{\mathbb{I}}}

\newcommand{\mireal}[1][]{
  \ifx\relax#1\relax%
    \II(\mione \,; \mitwo)%
  \else%
    \II(\mione \,; \mitwo\mid #1)%
  \fi
}

\newenvironment{tbox}{\begin{tcolorbox}[
		enlarge top by=5pt,
		enlarge bottom by=5pt,
		 breakable,
		 boxsep=0pt,
                  left=4pt,
                  right=4pt,
                  top=10pt,
                  arc=0pt,
                  boxrule=1pt,toprule=1pt,
                  colback=white
                  ]
	}
{\end{tcolorbox}}

\title{Fully Dynamic Algorithms for Coloring Triangle-Free Graphs}
\author{Sepehr Assadi\footnote{(sepehr@assadi.info) School of Computer Science, University of Waterloo. 
Supported in part by an NSERC
Discovery Grant and a Faculty of Math Research Chair grant. \smallskip} \\ {\small University of Waterloo} \and 
\and Helia Yazdanyar\footnote{(hyazdanyar@uwaterloo.ca)  School of Computer Science, University of Waterloo. Supported in part by a Cheriton Scholarship from School of Computer Science and SA's NSERC Discovery Grant.} \\ {\small University of Waterloo}
}

\date{}

\begin{document}

	\maketitle
	
	\pagenumbering{roman}
	

\begin{abstract}

\bigskip

A celebrated result of Johansson in graph theory states that every \textbf{triangle-free} graph of maximum degree $\Delta$ can be properly colored with $O(\Delta/\ln\Delta)$ colors, improving upon
the ``greedy bound'' of $\Delta+1$ coloring in general graphs. This coloring can also be found in polynomial time. 

\medskip

We present an algorithm for maintaining an $O(\Delta/\ln\Delta)$ coloring of a dynamically changing triangle-free graph that undergoes edge insertions and deletions. The algorithm is 
randomized and on $n$-vertex graphs has amortized update time of $\Delta^{o(1)}\log{n}$ per update with high probability, even against an adaptive adversary. 

\medskip

A key to the analysis of our algorithm is an application of the \textbf{entropy compression} method that to our knowledge is new in the context of dynamic algorithms. 
This technique appears general and is likely to find other applications in dynamic problems and thus can be of its own independent interest. 

\end{abstract}

	\bigskip
	
	\setcounter{tocdepth}{3}
	\tableofcontents
	
	\clearpage
	
		\setcounter{page}{1}
	\pagenumbering{arabic} 
	

\section{Introduction}\label{sec:intro}

Graph coloring is one of the most studied topics in graph theory with a wide range of applications in computer science. Let $G=(V,E)$ be an $n$-vertex undirected graph with maximum degree $\Delta$. 
For any integer $c \geq 1$, a (proper) $c$-coloring of $G$ is an assignment of colors from $\set{1,\ldots,c}$ to the vertices of the graph so that no edge is monochromatic. 
It is easy to see that every graph admits a $(\Delta+1)$ coloring which can be found via a simple greedy algorithm in linear time: color the vertices one by one in an arbitrary order and for each one, choose a color not assigned to any of its already colored
neighbors, which always exist by the pigeonhole principle. 

Nevertheless, one can almost always find a coloring with fewer colors. Brooks' theorem~\cite{Brooks41} states that the only (connected) graphs that do need $\Delta+1$ colors are
$(\Delta+1)$-cliques (and odd cycles only when $\Delta=2$); all other graphs can be colored with $\Delta$ colors. As already pointed out by Vizing in~\cite{Vizing68}, this result is just the tip of the iceberg and there is a vast body of work in graph theory with the general theme of 
coloring graphs with ``fewer'' than $\Delta$ colors based on how ``far'' they get from being a clique; see, e.g.,~\cite{Kim95,Johansson96,Reed98,Reed99,AlonKS99,MolloyR02,MolloyR14,BansalGG15,Molloy18,BonamyKNP22}.  A central result here is a seminal theorem of Johansson~\cite{Johansson96} that proves that triangle-free graphs are $O(\Delta/\ln\Delta)$ colorable. This bound is asymptotically optimal and has been since sharpened to $(1+o(1))\Delta/\ln{\Delta}$ 
colors by Molloy~\cite{Molloy18} (see also~\cite{Jamall11,PettieS13,Bernshteyn19,AchlioptasIS19}).

We study coloring of triangle-free graphs in the dynamic setting, wherein the graph undergoes edge insertions and deletions and we want to maintain a proper coloring after every update.
There has been extensive interest in 
 dynamic algorithms for graph coloring problems and specifically $(\Delta+1)$ coloring~\cite{BhattacharyaCHN18,BhattacharyaGKL22,HenzingerP22,BehnezhadRW25,FlinH25}. The state of the art for $(\Delta+1)$ coloring are randomized algorithms with 
 amortized expected update time of $O(1)$ against oblivious adversaries~\cite{BhattacharyaGKL22,HenzingerP22} and $\Ot(n^{2/3})$ against adaptive adversaries~\cite{FlinH25}\footnote{An oblivious adversary fixes the set of updates to the graph 
before presenting them to the algorithm, whereas an adaptive adversary may present each update to the graph based on the output of the algorithm on previous updates; see~\cite{BeimelKMNSS22,BernsteinBKS25,BehnezhadRW25} for more discussion
on the different types of adversary.}. Other variants of dynamic graph coloring have also been studied including arboricity-dependent coloring~\cite{ChristiansenR22,ChristiansenNR23} and edge coloring~\cite{DuanHZ19,ChristiansenRV24,BhattacharyaCPS24}. However, 
to our knowledge, there have been no prior work on dynamic algorithms for coloring triangle-free graphs. This is in contrast to related models such as static~\cite{Jamall11,Molloy18,AchlioptasIS19}, parallel~\cite{GrableP98}, distributed~\cite{PettieS13}, streaming~\cite{AlonA20} and sublinear time~\cite{AlonA20} algorithms. 
We address this gap in our paper. 

Before getting to our results, let us mention what existing static algorithms can imply for the problem in dynamic graphs. 
Previously,~\cite{Jamall11,Molloy18} provided $\Ot(n\Delta^2)$ time algorithms\footnote{Throughout, we use the standard notation $\Ot(f) := O(f \cdot \poly\log{f})$ to suppress logarithmic factors.} for $O(\Delta/\ln{\Delta})$ coloring of triangle-free \emph{static} graphs. By running this algorithm after each $\Theta(\Delta/\ln{\Delta})$ updates and using new colors in the meantime to keep a proper coloring, this implies a dynamic algorithm with $\Ot(n\Delta)$ update time. Using the sublinear time algorithm of~\cite{AlonA20} in place of the static algorithms can further reduce this to $\min{(n^{1+o(1)},n^{2+o(1)}/\Delta^2)}$ update time. 
It is worth comparing this state of affairs with the trivial dynamic algorithm for $(\Delta+1)$ coloring that recolors at most one vertex per update, by spending $O(\Delta)$ time to iterate over its neighbors and find an available color. 
This strategy however does not work for $O(\Delta/\ln\Delta)$ coloring of triangle-free graphs given there is no reason why one can greedily extend such a coloring to a new vertex. 


\subsection{Our Contributions}\label{sec:results}

We present a randomized algorithm for maintaining an $O(\Delta/\ln \Delta)$ coloring of triangle-free graphs in $\Delta^{o(1)}\log{n}$ amortized update time. 

\begin{result}\label{res:main}
	There is an algorithm that for any {constant} $\gamma \in (0,1)$, 
	maintains an $O_{\gamma}(\Delta/\ln{\Delta})$ coloring of any $n$-vertex fully dynamic triangle-free graph  in $O(\Delta^{\gamma}\log{n})$ amortized update time. 
	The algorithm is randomized and works with high probability against an adaptive adversary (against an oblivious adversary, the update time improves to $O(\Delta^{\gamma})$ time). 
\end{result}
\Cref{res:main} provides the first non-trivial dynamic algorithm for coloring triangle-free graphs, and goes way below the benchmark of re-running static algorithms after a fixed number of updates. Indeed, by letting $\gamma \rightarrow 0$, 
we can reduce the update time to $\Delta^{o(1)}\log{n}$ time while maintaining an $O(\Delta/\ln{\Delta})$ coloring of the graph. It is worth contrasting this bound with the state-of-the-art bound of $\Ot(n^{2/3})$ amortized update time
for $(\Delta+1)$ coloring of arbitrary graphs. 

A key component of our main algorithm in~\Cref{res:main} is a dynamic algorithm for recoloring the graph after each update analogous to the trivial $O(\Delta)$ update time algorithm for $(\Delta+1)$ coloring. Unlike greedy coloring, 
our algorithm is based on a local search strategy, that \emph{on average}, requires spending $O(\Delta^2\log{n})$ time per update. 

\begin{result}\label{res:Delta}
	For any constant $\eps \in (0,1)$, there is an algorithm for maintaining a $(1+\eps)\Delta/\ln{\Delta}$ coloring of $n$-vertex fully dynamic triangle-free graphs with maximum degree $\Delta$ (sufficiently large as a function of $\eps$). 
	The algorithm is randomized and with high probability has amortized update time of $O(\Delta^2\log{n})$ even against adaptive adversaries (against an oblivious adversary, the update time improves to $O(\Delta^{2})$ time). 
\end{result}

\Cref{res:Delta} can be seen as a \emph{dynamic} version of the Molloy's breakthrough~\cite{Molloy18} that proves triangle-free graphs are $(1+o(1))\Delta/\ln{\Delta}$ colorable, sharpening the leading constants 
in the long line of work on this problem~\cite{Johansson96,MolloyR02,Jamall11,PettieS13,BansalGG15}. Our algorithm is inspired by Molloy's proof of using the \textbf{entropy compression} method,
originally introduced by Moser~\cite{Moser09} and Moser and Tardos~\cite{MoserT10} in their seminal work on algorithmic Lovasz Local Lemma. 

\paragraph{Our techniques.} Employing Molloy's entropy compression directly only leads to a dynamic algorithm
with $O(n\Delta)$ update time (see~\cite{Molloy18}). However, we show how to make the entropy compression \emph{dynamic}, which to our knowledge, is new in this context (for readers familiar with this technique, we achieve this by considering the interaction of the randomized dynamic algorithm and adversary \emph{jointly} as a compression algorithm, logging the ``standard'' parts of the technique combined with dynamic changes to the graph as part of the compression, and design a recovery algorithm that given this log can recover the original random bits despite the dynamic changes to the graph; see~\Cref{sec:alg-partial-color} for more details). 
Beside this, Molloy uses entropy compression to reduce the problem to a list-coloring problem of Reed~\cite{Reed99c} (see also~\cite{ReedS02}). For our purpose, this reduction is not sufficient as we do not know how to solve Reed's problem dynamically. Instead, we use (dynamic) entropy compression to reduce the problem to the $(\Delta+1)$-coloring problem and then can use the $O(\Delta)$ update time dynamic algorithm for the problem. 

To prove~\Cref{res:main} from~\Cref{res:Delta}, we borrow an idea originally due to~\cite{BhattacharyaCHN18} for $(1+o(1))\Delta$ coloring of arbitrary dynamic graphs. The general approach is to maintain a hierarchy of vertex partitions of the graph, and ensure that at each level of the hierarchy, each vertex belongs to a partition with ``few'' other neighbors in its own part (roughly speaking, we would like to \emph{dynamically} partition vertices into $b$ groups in each level, so each vertex $v$ have $\approx \deg(v)/b$ neighbors in its group). We implement this by designing a family of recursive algorithms with improved update times whose base case is our algorithm in~\Cref{res:Delta}. This recursive approach, as opposed
to the ``iterative'' argument of~\cite{BhattacharyaCHN18}, allows us to better control the parameters involved and extend the argument to $O(\Delta/\ln{\Delta})$ coloring, which in particular, is a \emph{non-linear} function of vertex degrees. 

\paragraph{Perspective: dynamic entropy compression.} Entropy compression, as argued by~\cite{Tao09}---who also coined this term for the techniques of~\cite{Moser09,MoserT10}---is a unique type of \emph{potential function} argument. Suppose we want a coloring of a graph that satisfy a certain ``local'' constraint, e.g., the constraint of each vertex depends on colors of its neighbors. Then, we start with an arbitrary coloring of the graph, and then as long as we find a vertex $v$ that violates its constraint, we \emph{randomly} and \emph{locally} change the colors of vertices around it. Suppose we could even say that ``with large probability'' this random change
fixes the constraint for the vertex $v$. Can we also argue that repeating this process enough number of times fixes \emph{all} constraints? The challenge is that fixing constraint of a single vertex may lead to many other constraints for other vertices to be violated now. 
Entropy compression is precisely a method to keep track of the ``progress'' made by this algorithm in a potential-function type argument. At its core, it argues that if we can run this process for a large number of times, we will be able to compress a 
truly random string to a summary of much smaller length, which is information-theoretically impossible.

Potential functions appear quite frequently in the analysis of dynamic graph algorithms. Yet, to our knowledge, entropy compression is curiously omitted from the standard toolkit in this literature. 
Our work shows that this technique can be naturally extended to analyze dynamic graph algorithms as well. We especially find this interesting as it allows for analyzing \emph{randomized} algorithms against \emph{adaptive} adversaries
in a new way. Most adaptive dynamic algorithms use randomization to run \emph{fast sublinear time} algorithms once every fixed number of updates (this allows them to address the correlation between randomness 
and adaptive adversary's update by using \emph{fresh} randomness in the analysis); see, e.g., the adaptive dynamic algorithms for $(\Delta+1)$ coloring~\cite{BehnezhadRW25,FlinH25}. 
In contrast, dynamic entropy compression allows us to \emph{continuously} use randomness in the algorithm and still analyze its performance against an adaptive adversary, similar to how
one typically does with a potential function in a deterministic algorithm (the only other technique of this nature we are aware
of is \emph{proactive resampling} introduced in~\cite{BernsteinBGNSS022,BhattacharyaSS22} although the similarities here are more conceptual than technical). We believe that this viewpoint is an important (non-technical) contribution of our paper as it may pave the way for obtaining more eﬃcient dynamic algorithms for other fundamental graph problems as well. 
 
\paragraph{Independent and concurrent works.} Parallel to us,~\cite{HaeuplerP26} further developed the proactive resampling technique of~\cite{BernsteinBGNSS022,BhattacharyaSS22}. Among many other interesting results, 
they also presented a fully dynamic $O(\Delta/\ln{\Delta})$ coloring algorithm for triangle free graphs with update time of $n^{1/2+o(1)}$ by ``dynamizing'' the sublinear time algorithm of~\cite{AlonA20}. Technique-wise, our work and~\cite{HaeuplerP26}
appear to be entirely disjoint. 

More relatedly,~\cite{HaeuplerMRST26} also observed that entropy compression gives a recipe for designing dynamic algorithms and used this to prove a wide class of local search algorithms extend as is to the fully dynamic setting. 
As a corollary of their techniques, they also obtained an entropy compression based 
dynamic algorithm for coloring triangle-free graphs similar to our~\Cref{res:Delta}. 
Their results are quite more general than ours in scope and span various applications of algorithmic Lovasz Local Lemma and beyond. But, for the specific problem at hand, their general framework gives an $\Ot(\Delta^3)$ update time algorithm for $6\Delta/\ln{\Delta}$ coloring of triangle-free graphs, which is quantitatively weaker than~\Cref{res:Delta} (the extension
of our results to $\Delta^{o(1)}\log{n}$ update time in~\Cref{res:main} is using a separate set of techniques that is missing from~\cite{HaeuplerMRST26}). At a technical-level, our dynamic entropy compression works by
``logging dynamic updates'' whereas~\cite{HaeuplerMRST26} do a union bound over all possible choices of adversary, which are different perspectives on the same general idea.



\section{Preliminaries}\label{sec:prelim}

\paragraph{Notation.} For any integer $b \geq 1$, define $[b] := \set{1,\cdots,b}$. 
For a graph $G=(V,E)$, we use $N(v)$ to denote the neighbors of $v \in V$, $\deg(v)=\card{N(v)}$ to denote the degree of $v$, and $B(v,r)$ to denote the ball of radius $r$ around $v$, i.e., all vertices that are at distance at most $r$ from $v$. 
A \textbf{partial coloring} of a graph $G=(V,E)$ is a function
$\varphi : V \to \mathcal{C} \cup \{\perp\}$, where $\mathcal{C}$ is a set of colors and
$\perp$ denotes an uncolored vertex, also referred to as a \emph{free} vertex.
A partial coloring $\varphi_t$ is said to be \textbf{proper} if no edge is monochromatic, however vertices assigned color $\perp$ are exempt, namely, are allowed to be adjacent to each other.

\paragraph{Dynamic graphs.} We denote the static vertex set of the underlying $n$-vertex graph by $V := [n]$.  
Let $\mathcal{G} = (G_0, G_1,\cdots, G_t, \cdots)$ be the sequence of dynamic graphs presented to the algorithm: the initial graph $G_0$ is empty and each graph $G_t$ is obtained from the previous graph $G_{t-1}$ by inserting 
or deleting a single edge $e_t = (u_t,v_t)$. We use $G_t := (V,E_t)$ to denote the graph at the $t$-th step; $N_t(v)$, $\deg_t(v)$, and $B_t(v,r)$ for any $v \in V$ and $r \geq 0$ are defined analogously with respect to $G_t$. 

As is common in this context, we assume $\Delta$ is known a priori and every vertex has degree at most $\Delta$ throughout the updates. 
We also assume the number of updates is $\Omega(n)$ as otherwise, there will be vertices that do not even receive a single update and are hence redundant\footnote{This is relevant when
we would like to achieve a \emph{`with high probability'} guarantee for some aspect of the algorithm (say, runtime). If the updates are much smaller than $n$, then one naturally cannot expect a probability bound based on $n$ to hold (think of padding an
arbitrary graph with ``near infinitely many'' vertices).}. 

An \textbf{oblivious} adversary fixes the choices of graphs in $\mathcal{G}$ before presenting them to the algorithm whereas an \textbf{adaptive} adversary can choose each graph $G_t$ based on the output of the algorithm
on prior graphs $G_1,\ldots,G_{t-1}$ (in our case, we assume even the prior random coin tosses of the algorithm are available to the adversary). 
This adversary can be modeled as a game between the dynamic algorithm and an oracle $\mathcal{O}_A$ that at each time step $t$ receives the random coins used along with the current coloring and outputs the next update.

Finally, the \textbf{recourse} of a dynamic graph coloring algorithm is the worst-case number of vertices whose colors have changed by the algorithm after each single update.

\paragraph{Balls into bins.} We use the following variant of balls and bins experiments. Consider a set of $d$ balls and $b$ bins $\mathcal{B} = \{B_1, \ldots, B_b\}$. For each ball $i \in [d]$, we are given a subset of bins $S_i \subseteq \mathcal{B}$. Each ball $i$ is placed into one of the bins in $S_i$ uniformly at random.

\begin{proposition}[\!\!\cite{PS97,Molloy18}]\label{lem:concentration-ineq-negative-associated}
   For $i \in [b]$, let $X_i$ be the indicator random variable that bin $B_i$ is non-empty in the balls-into-bins experiment above, and define $X = \sum_{i=1}^b X_i$. Then, for any $0 < t \leq \expect{X}$, we have:
    \[
    \prob{X<\expect{X}-t}< \exp\left(-\frac{t^2}{2\expect{X}}\right).
    \]
\end{proposition}
\begin{proof}[Proof sketch]
     For each $i \in [b]$, let $Y_i = 1 - X_i$ and let $N_i$ be the random variable denoting the number of balls placed in bin $B_i$. As $(N_i:i\in [b])$ has multinomial distribution, the collection $\{N_i\}_{i\in [b]}$ is negatively associated (NA)~\cite[Example~3.1]{joagProschan1983}.

    Note that $Y_i$ is an indicator variable that is equal to $1$ if $N_i= 0$ and $0$ otherwise. By the properties of NA~\cite{joagProschan1983}, any non-increasing functions of NA variables is also NA. Therefore, $\{Y_i\}_{i \in [q]}$ is NA. Applying the concentration bound for NA variables from~\cite[Lemma 3b]{Molloy18} yields the desired bound. 
\end{proof}

\paragraph{Concentration Inequality.}
We use the following standard version of Chernoff bound. 
\begin{proposition}[Chernoff-Hoeffding bound; cf.~\cite{MolloyR02}]\label{prop:chernoff}
	Let $X_1,\ldots,X_n$ be $n$ independent indicator random variables. Define $X:= \sum_{i=1}^{n} X_i$. 
	For any $\delta \in (0,1)$, and $\expect{X} \leq \mu$ we have: 
	\begin{align*}
		\Pr\paren{X > (1+\delta) \cdot \mu} \leq \exp\paren{-\frac{\delta^2 \cdot\mu}{3}}.
	\end{align*} 
\end{proposition}

\paragraph{Entropy compression.}
The following lemma states that we cannot compress random bits efficiently, which forms the key concept behind the entropy compression method. 
This is a pretty standard result and we provide its elementary proof for completeness.

\begin{lemma}\label{lem:compression}
    Let $\delta \in (0,1)$ and $f : \{0,1\}^m \rightarrow \{0,1\}^t$ be any fixed function with the following property: if we
    sample $x$ uniformly at random from $\{0,1\}^m$, we can recover $x$ from $f(x)$ with probability at least $\delta$. Then, we should have 
    \[
    t\geq m - \log(1/\delta).
    \]
\end{lemma}
\begin{proof}
    Let $X \subseteq \{0, 1\}^m$ be the set of strings that are correctly recovered by $f$. For any two distinct $x, y \in X$, it must be that $f(x) \neq f(y)$; otherwise, we cannot uniquely recover $x$ from function $f$. Thus, $f:X\rightarrow \{0, 1\}^t$ must be injective, which implies that $\abs{X}\leq 2^t$.
    
    We also know that $x\in \{0, 1\}^m$  can be recovered from $f(x)$ with probability at least $\delta$, therefore $\abs{X}\geq \delta\cdot 2^m$. Hence, 
    \[
    2^t\geq \abs{X} \geq \delta\cdot 2^m, 
    \]
    resulting in the desired $t\geq m - \log(1/\delta)$, concluding the proof.
\end{proof}


\newcommand{\colort}[1]{\ensuremath{\textnormal{\texttt{color}}_t(#1)}}
\newcommand{\availt}[1]{\ensuremath{\textnormal{\texttt{Avail}}_t(#1)}}
\newcommand{\freet}[1]{\ensuremath{\textnormal{\texttt{Free}}_t(#1)}}

\newcommand{\CC}{\ensuremath{\mathcal{C}}}

\newcommand{\update}{\ensuremath{\textnormal{\textsc{Update}}}\xspace}
\newcommand{\pcolort}{\ensuremath{\varphi_t}\xspace}

\newcommand{\partialcolor}{\ensuremath{\textnormal{\textsc{PartialColor}}}\xspace}

\newcommand{\fix}{\ensuremath{\textnormal{\textsc{Fix}}}\xspace}
\newcommand{\fixA}{\ensuremath{\textnormal{\textsc{Fix}}_{\texttt{Avail}}}\xspace}
\newcommand{\fixF}{\ensuremath{\textnormal{\textsc{Fix}}_{\texttt{Free}}}\xspace}
\newcommand{\resample}{\ensuremath{\textnormal{\textsc{Resample}}}\xspace}
\newcommand{\verify}{\ensuremath{\textnormal{\textsc{Verify}}}\xspace}

\section{An $O(\Delta^2\log{n})$ Update Time Algorithm}\label{sec:poly-Delta}

We now state our main technical contribution, formalizing~\Cref{res:Delta}. 

\begin{theorem}[Formalizing~\Cref{res:Delta}]\label{thm:poly-Delta}
	Let $\eps \in (0,1)$ be a fixed constant and $\Delta_0$ be sufficiently large as a function of $\eps$. There is an algorithm that maintains a 
	$
		(1+\eps) \cdot \Delta / \ln{\Delta}
	$
	coloring of any fully dynamic $n$-vertex triangle-free graph of maximum degree $\Delta \geq \Delta_0$. The algorithm is randomized and with high probability has amortized update time of $O(\Delta^2\log{n})$ against adaptive adversaries and $O(\Delta^2)$ against oblivious adversaries.  
\end{theorem}

We first describe the maintained invariants and the data structures supporting dynamic updates. Next, we introduce an algorithm for maintaining a proper partial coloring of the vertices and analyze its amortized update time, which forms the core technical component of our approach. We then show how this partial coloring can be extended to a full coloring, completing the proof of \Cref{thm:poly-Delta}.

\subsection{Setup}
Recall that $\Delta$ is a known upper bound on the maximum degree of the graph during the updates. Throughout this section, we fix $\eps > 0$ to be any arbitrarily small \emph{constant}, and use the following two palettes of colors: 
\[
	\CC_1 := \set{1,\cdots,q} \qquad \text{and} \qquad \CC_2 := \set{q+1,\ldots,(1+\eps) \cdot q} \qquad \text{where} \qquad q := \frac{(1+\eps) \cdot \Delta}{\ln{\Delta}}. 
\]

\noindent
Here $\abs{\CC_1}=q$ and $\abs{\CC_2}=\eps q$. We aim to use $\CC_1$ to color most of the vertices and use $\CC_2$ to complete the coloring through a cleanup algorithm.

For every $v \in V$, at any time $t$ during the sequence of graphs $(G_0,G_1,\cdots,G_t,\cdots)$, we define: 
\begin{itemize}
	\item $\colort{v}$: the color in $\CC_1 \cup \CC_2$ assigned to $v$ at time $t$; our goal is to keep $\colort{\cdot}: V \rightarrow \CC_1 \cup \CC_2$ as a proper coloring of $G_t$. 
\end{itemize}
The main part of the argument is a fully dynamic algorithm $\partialcolor$ that maintains a \emph{partial} coloring $\pcolort: V \rightarrow \CC_1 \cup \set{\perp}$ with several important properties for us. Specifically, we define, for any vertex $v \in V$ and time $t$: \begin{itemize}
	\item $\availt{v}$: colors in $\CC_1$ that are \emph{available} to $v$ under $\pcolort$, i.e., have not been assigned to any neighbor of $v$ in $N_t(v)$; further, set $a_t(v) := \card{\availt{v}}$. 
	\item $\freet{v}$: neighbors of $v$ in $N_t(v)$ that are \emph{free}, i.e., have been marked with $\perp$; further, set $f_t(v) := \card{\freet{v}}$. 
\end{itemize}

\begin{lemma}\label{lem:partial}
	$\partialcolor$ maintains a partial coloring $\pcolort: V \rightarrow \CC_1 \cup \set{\perp}$ after each update at time $t$ such that for all $v \in V$: 
	\begin{equation}\label{eq:constraints}
			a_t(v) \geq \Delta^{\eps/2} \qquad \text{and} \qquad f_t(v) \leq \frac{\eps q}{2}.
	\end{equation}
	Moreover, after at least $n$ many updates, the amortized recourse and update time of $\partialcolor$ with high probability are, respectively,
	\begin{itemize}
		\item $O(\Delta \cdot \log{n})$ and $O(\Delta^2 \cdot \log{n})$ against adaptive adversary, 
		\item $O(\Delta)$ and $O(\Delta^2)$ against oblivious adversary. 
	\end{itemize}
\end{lemma}

We describe $\partialcolor$ in the following subsection and then provide its analysis. We then use~\Cref{lem:partial} to conclude the proof of~\Cref{thm:poly-Delta}. 

\subsection{The $\partialcolor$ Algorithm}\label{sec:alg-partial-color}

\partialcolor is a simple local search algorithm:
it starts with a graph $G_{t-1}$ that satisfies the constraints of~\Cref{lem:partial}, receives an update $e_t = (u_t,v_t)$, and simply tries to ``locally fix'' any vertex that no longer 
satisfies the guarantees of~\Cref{eq:constraints} by resampling the colors of its neighbors and recursing. Formally, the algorithm is as follows: 

\begin{Algorithm}\label{alg:partial-color}
$\partialcolor(e_t)$: Called when $e_t = (u_t,v_t)$ is updated in $G_{t-1}$ to obtain $G_t$. 
\begin{enumerate}
	\item Update the data structures $N_t(\cdot)$, $\availt{\cdot}$, and $\freet{\cdot}$ for vertices $u_t,v_t$. 
	\item If $e_t$ is inserted and $\pcolort(u_t) = \pcolort(v_t)$, set $\pcolort(u_t) = \perp$. 
	\item while there is $w \in \set{u_t,v_t} \cup N(u_t)$ with $a_t(w) < \Delta^{\eps/2}$ or $f_t(w) > \eps q/2$, run $\fix(w)$. 
\end{enumerate}


\noindent
\underline{\textbf{Subroutine} $\fix(u)$:}\label{subroutine:resample}
\begin{enumerate}
	\item\label{line:resample-1} If $a_t(u) < \Delta^{\eps/2}$, run $\resample(u)$. 
	\item If $f_t(u) > \eps q/2$, then while there is $v \in N_t(u)$ with $a_t(v) < \Delta^{\eps/2}$, run $\fix(v)$. 
	\item\label{line:resample-2} If $f_t(u) > \eps q/2$ still, run $\resample(u)$. 
\end{enumerate}

\noindent
\underline{\textbf{Subroutine} $\resample(u)$:} 
\begin{enumerate}
	\item For $z \in N(u)$: resample $\pcolort{(z)}$ uniformly at random from $\availt{z} \cup \set{\perp}$; update $\availt{\cdot}$ and $\freet{\cdot}$ for neighbors of $z$ in $N_t(z)$. 
	\item While there is $w \in B_t(u,2)$ with $a_t(w) < \Delta^{\eps/2}$ or $f_t(w) > \eps q/2$, run $\fix(w)$. 
\end{enumerate}
\end{Algorithm}

$\partialcolor$ aims to eliminate all unsatisfied constraints in \Cref{eq:constraints} by invoking the subroutine $\fix$ on any vertex that violates a constraint. Note that a priori, it is not at all clear that the subroutine $\fix$ 
would ever terminate even. We will get to that in the main part of the proof. But before that, the following observation shows that each \emph{terminating} call to $\fix$ strictly reduces the number of unsatisfied constraints.

\begin{observation}\label{obs:flaw-resolved}
	For any $u \in V$, assuming $\fix(u)$ terminates, the partial coloring maintained by the algorithm satisfies both constraints of~\Cref{eq:constraints} for $u$. Moreover, $\fix(u)$, upon termination, does not break any constraint that was already satisfied.
\end{observation}
\begin{proof}
	Whenever $u$ does not satisfy a constraint, subroutine $\fix(u)$, assuming other calls to \fix triggered along the way terminate, triggers a call to $\resample(u)$.
Since $u$ is in its own 2-hop neighbourhood, \resample triggers calls to $\fix(u)$ again. Therefore, the procedure $\fix(u)$ does not exit its while-loops as long as any of the constraints of $u$ remain unsatisfied.

Observe that a call to $\fix(v)$ for any $v \in V$ can generate new violated constraints
only when it executes $\resample(v)$, which changes the colors of vertices in $N_t(v)$.
Such color changes can only affect the constraints of vertices in $B_t(v,2)$.
However, $\resample(v)$ explicitly invokes $\fix(w)$ for every vertex
$w \in B_t(v,2)$ whose constraints may have been affected.

Consider an initial call to $\fix(u)$. Whenever a recoloring causes a constraint
violation at some vertex $v$, the algorithm invokes $\fix(v)$.
Assuming inductively that each invocation of $\fix(v)$ restores all constraints of $v$
before terminating, it follows that all constraint violations caused by the execution
of $\fix(u)$ are eventually repaired. Therefore, upon termination of $\fix(u)$,
no previously satisfied constraints will be violated.
\end{proof}

To bound the update time of $\partialcolor$ we use the entropy compression technique. We argue that the algorithm effectively compresses the random bits used to color the vertices as part of the \resample subroutine. Consequently, if the algorithm runs for too long, it would compress the randomness beyond what is allowed by \Cref{lem:compression}, leading to a contradiction. The key to this step is the following probabilistic lemma. A similar lemma also appears in~\cite{Molloy18} 
(although the second part is entirely different given we use a different ``flaw'' in the language of~\cite{Molloy18}).

\begin{lemma}\label{lem:prob-bound-bad-events}
	At any point in time in the algorithm: 
	\begin{enumerate}
	\item when calling $\resample(u)$ in Line~\eqref{line:resample-1} of $\fix(u)$:  
	\[
		\prob{a_t(u) < \Delta^{\eps/2}}\leq 2^{-\frac{\Delta^{\eps/2}}{100}};
	\]
	\item when calling $\resample(u)$ in Line~\eqref{line:resample-2} of $\fix(u)$:
	\[
	\prob{f_t(u)>\frac{\eps q}{2}}\leq 2^{-\frac{\eps\cdot \Delta}{50\ln \Delta}}.
	\]
	\end{enumerate}
\end{lemma}
\begin{proof} We prove each part separately.

	\textbf{Part 1.}
	The proof is by bounding $\expect{a_t(u)}$ and applying a concentration bound.
	Recall that a color $c \in \CC_1$ is available for $u$ at time $t$ iff no neighbor $v \in N_t(u)$ samples color $c$. Since each neighbor $v$ chooses a color uniformly at random from the set $\availt{v} \cup \{\perp\}$, the probability that $v$ 
	chooses $c$ is exactly $1/(a_t(v) + 1)$ if $c \in \availt{v}$, and $0$ otherwise. 
	By the independence of these color choices (as $G$ is triangle-free), the probability that $c \in \availt{u}$ is:
	\[
	\Pr\left(c\in \availt{u}\right) = \prod\limits_{\substack{v\in N_t(u) \\ c\in \availt{v}}} \left(1-\frac{1}{a_t(v)+1}\right).
	\]
	Therefore, we have
	\begin{align*}
		\expect{a_t(u)} 
		&= \sum_{c\in \CC_1} \prod\limits_{\substack{v\in N_t(u) \\ c\in \availt{v}}} \left(1-\frac{1}{a_t(v)+1}\right)\\
		&\geq \sum_{c\in \CC_1} \prod\limits_{\substack{v\in N_t(u) \\ c\in \availt{v}}} \exp\paren{-\frac{1}{a_t(v)}} \tag{as $1-1/(x+1) \geq e^{-1/x}$ for $x > 0$}\\
		&= \sum_{c\in \CC_1} \exp(-\!\sum\limits_{\substack{v\in N_t(u) \\ c\in \availt{v}}} \frac{1}{a_t(v)}) =q\cdot \Exp_{c \sim \CC_1}\left[e^{-\rho(c)}\right],
	\end{align*}
	where we define
	\[
	\rho(c):= \sum\limits_{\substack{v\in N_t(u) \\ c\in \availt{v}}} \frac{1}{a_t(v)} \quad \text{for each $c\in \CC_1$}. 
	\]
	We can see that 
	\begin{align*}\label{eq:expect-rho(c)}
		\Exp_{c\sim \CC_1}\left[\rho(c)\right] 
		= \sum_{c\in\CC_1}\frac{1}{q}\cdot \rho(c) 
		=  \sum_{c\in\CC_1}\frac{1}{q}\cdot \sum\limits_{\substack{v\in N_t(u) \\ c\in \availt{v}}} \frac{1}{a_t(v)} 
		 = \frac{1}{q}\cdot \sum_{v\in N_t(u)} \sum_{c\in \availt{v}}\frac{1}{a_t(v)} \leq \frac{\Delta}{q}.
	\end{align*}
	This is because each $v\in N_t(u)$ has exactly $a_t(v)$ many available colors and $\deg_t(u)\leq \Delta$. Therefore, using Jensen's inequality and convexity of $e^{-x}$ we get 
	\begin{equation}\label{eq:expect-avail}
		\expect{a_t(u)}  \geq q\cdot \exp(-\Exp[\rho(c)]) \geq q\cdot \exp(-\frac{\Delta}{q}) = \frac{(1+\eps)\cdot \Delta}{\ln \Delta} \cdot \exp(-\frac{\ln\Delta}{1+\eps})\geq 2\Delta^{\eps/2}, 
	\end{equation}
	where the last inequality is true because $\Delta$ is sufficiently large. 

	We now prove a concentration bound for $a_t(u)$. The sampling of colors here corresponds to the balls-into-bins framework described in \Cref{lem:concentration-ineq-negative-associated}, where the colors in $\CC_1$ serve as the bins and the neighbors
	 $v \in N_t(u)$ act as the balls. Specifically, for each ball $i \in [\deg_t(v)]$, the set of allowed bins $S_i$ in the experiment is $\availt{v_i} \cup \set{\perp}$. 
	Therefore, we can use the same concentration inequality as in \Cref{lem:concentration-ineq-negative-associated} to bound $a_t(u)$. Using \Cref{eq:expect-avail}, we get
	\begin{align*}
	\prob{a_t(u) < \Delta^{\eps/2}} \leq \prob{a_t(u)< \frac{\expect{a_t(u)}}{2}} 
	\leq \exp{(-\frac{\expect{a_t(u)}}{12})}
	\leq \exp{(-\frac{\Delta^{\eps/2}}{12})}
	\leq 2^{-\frac{\Delta^{\eps/2}}{100}},  
	\end{align*}
	concluding the proof of the first part. 

	\noindent
	\textbf{Part 2.} 
	For any $v \in N_t(u)$, let $Z_v$ be the indicator random variable for $v$ picking $\perp$. We have
	\[
		\expect{f_t(u)}=\sum_{v\in N_t(u)} \expect{Z_v} = \sum_{v \in N_t(u)} \frac{1}{a_t(v)+1}.
	\] 
	By~\Cref{obs:flaw-resolved} and since we now consider $\resample(u)$ called in Line~\eqref{line:resample-2}, we know that $a_t(v) \geq \Delta^{\eps/2}$ for all $v \in N_t(u)$. 
	Hence, we have, 
	\begin{align*}
		\expect{f_t(u)} \leq \Delta \cdot \frac{1}{\Delta^{\eps/2}} \ll \frac{\eps q}{4}. 
	\end{align*}
	Since the random variables $\{Z_u\}_{u\in N_t(v)}$ are independent, by \Cref{prop:chernoff}, we have
	\begin{align*}
		\prob{f_t(v) \geq \eps q/2} 
		&\leq \exp(-\frac{\eps q}{12}) \leq 2^{-\frac{\eps q}{50}} \leq 2^{-\frac{\eps\cdot \Delta}{50\ln \Delta}},
	\end{align*}
concluding the proof of the second concentration inequality.
\end{proof}

We now define a compression scheme for compressing random bits used by the algorithm over a \emph{given} sequence of updates $(G_0,G_1,\ldots,G_T)$ for some $T \geq n$. 
In the following, we focus on the case when the updates are done by an \emph{adaptive} adversary (which is the more interesting case). We will then briefly describe how to handle oblivious adversaries more efficiently. 

\newcommand{\All}{\ensuremath{\textnormal{All}}\xspace}
\newcommand{\Bad}{\ensuremath{\textnormal{Bad}}\xspace}
\newcommand{\Bada}{\ensuremath{\textnormal{B}_a}\xspace}
\newcommand{\Badf}{\ensuremath{\textnormal{B}_f}\xspace}

To continue, we need some definitions. At any point of time during the algorithm, define: 
\begin{itemize}
	\item $\All_t(v)$: the set of all possible assignments of colors to neighbors $u$ of $v$ from $\availt{u} \cup \set{\perp}$. This way, we have, 
	\[
		\card{\All_t(v)} = \prod_{u \in N_t(v)} \paren{a_t(u)+1}. 
	\]
	\item $\Bad_t(v) \subseteq \All_t(v)$: those assignments that will lead to $v$ having $a_t(v) < \Delta^{\eps/2}$ or $f_t(v) > \eps q/2$. 
	
	(We shall note that the subscript $t$ these definitions is a bit misleading since these sets and parameters can vary many times during processing of a single update; however, to avoid cluttering the notation, 
	we stick with using subscript $t$ here.)
\end{itemize}

\paragraph{The compression scheme.} We maintain a \textit{log} of the algorithm on $(G_1,\ldots,G_T)$ that allows us to reconstruct all random bits used during its execution. At time step $t$, with update $e_t = (u_t, v_t)$, we record the following information:
\begin{enumerate}
	\item A single bit indicating the type of update (insertion or deletion), together with $2\log n$ bits specifying the edge $e_t$.
	\item A single bit indicating whether the algorithm terminates at the second line of $\partialcolor$.
	\item If the algorithm reaches Line 3 of $\partialcolor$, for any vertex $w \in N_t(u_t)$ for which $\fix(w)$ was called, we write $O(\log \Delta)$ bits to specify the vertex (we need $\log{\Delta}$ bits instead of $\log{n}$ as we specify this vertex as the $i$-th neighbor of 
	$u_t$ for some $i \in [\Delta]$). Additionally, we use $O(1)$ bits to indicate which constraint violation triggered the call to $\fix$.
	\item For each call to $\fix(u)$, the input $u$ is specified before this and we have also written down the type of constraint violation.
	This allows us to know whether $\fix(u)$ terminates in~Line~\eqref{line:resample-1} or Line~\eqref{line:resample-2}. 
	If $\fix(u)$ terminates in Line~\eqref{line:resample-1}, we consider the set $\Bad_t(u)$, namely, the assignment of colors to neighbors of $u$ at this point in time that will result in $u$ having a violated constraint. 
	Since $u$
	does have a violated constraint, the assignment of $\pcolort$ to $N_t(u)$ in this step belongs to $\Bad_t(u)$. We will write the index $\ell \in \card{\Bad_t(v)}$ to the log to specify this part of $\pcolort$. 
	
	If $\fix(u)$ goes beyond Line~\eqref{line:resample-1}, we again record each $v \in N_t(v)$ that we call $\fix(v)$ on using $O(\log{\Delta})$ bits. Then, we go to Line~\eqref{line:resample-2} and if $\fix(u)$ is called here, we write the index $\ell \in \card{\Bad_t(u)}$ 
	corresponding to the assignment of $\pcolort$ to $N_t(u)$ at this point (which we know belongs to $\Bad_t(u)$). 
	Finally, upon completing $\fix(v)$, we write $O(1)$ bits to mark the end of this subroutine. This process is repeated recursively within each call to $\fix$.

	\item At the end of all of the updates we also log the final choice of $\pcolort$ in $n\log{\Delta}$ bits. 
\end{enumerate}

This concludes the description of our compression scheme. Note that at no point in this process, we explicitly wrote down the random choices made by $\resample$ at the time they were made. Nonetheless, the following observation
allows us to argue that  we can recover all those random bits from the compression. 

\begin{observation}\label{obs:compression-scheme}
	Given the log described above for a sequence of $T$ updates and $S_T$ calls to $\resample$, we can recover all random bits used by $\resample$ in the execution of the algorithm. 
\end{observation}
\begin{proof}
	In the compression scheme, the entire sequence of graphs
	$G_1,\ldots,G_T$ is recorded explicitly.
	Moreover, every call to $\fix$ is fully specified in the log by recording
	the index of the vertex for which $\fix$ is invoked (in a ``local'' way by starting from $u_t$ of the corresponding update at time $t$, and using the stored $\log{\Delta}$-bit addresses to go to the specified neighbor and so on), together with the type of constraint violation that triggered the call.
	As a result, for any call to $\fix(v)$, we can uniquely reconstruct $v$ and the exact graph on which this call is executed.

	We reconstruct the random bits in the \emph{reverse} order.
	At time $T$, the compression explicitly records the final partial coloring $\varphi_T$. The log also specifies the last vertex $v$ whose call to $\fix$ triggered a call to $\resample$.
	Since $\resample(v)$ only randomizes the colors of vertices in the
	neighborhood of $v$, the coloring $\varphi_T$ restricted to this neighborhood uniquely determines the random bits used in this call.

	Furthermore, the compression includes an index pointing to the violated
	constraint that caused this call to \fix.
	This information allows us to reconstruct the coloring of the neighborhood of $v$ \emph{before} the call to $\resample(v)$, as the colors of all vertices outside this neighborhood remain unchanged during this step. Equivalently, from the fixed coloring outside the neighborhood we can reconstruct the set $\Bad(v)$ and hence recover the previous local coloring.

	Thus, we recover both the random bits used in the $S$-th call to $\resample$ and the partial coloring immediately preceding this call.
	By iterating this argument backwards for the remaining $S-1$ calls,
	we recover all random bits used by the algorithm.
\end{proof}

By~\Cref{obs:compression-scheme}, we can recover all random bits used by the algorithm on a sequence of updates using the compression scheme above. 
We now use~\Cref{lem:prob-bound-bad-events} to argue that this compression scheme does indeed compress the random bits: roughly speaking, this is because between all possible color assignments to neighborhood of a vertex, 
most of the resulting colorings do not violate the constraints in \Cref{eq:constraints} by~\Cref{lem:prob-bound-bad-events}. We formalize and use this in the following lemma to bound the runtime of $\partialcolor$. 

\begin{lemma}\label{lem:upper-bound-resample}
	Let $S_T$ be the random variable for the number of times $\resample$ is called after $T$ updates. Against an adaptive adversary, we have:
	\begin{align*}
		\prob{S_T\geq T\log{n}+n}\leq n^{-T}, 
	\end{align*}
	and against an oblivious adversary we have
	\begin{align*}
		\prob{S_T\geq T+n}\leq \Delta^{-T}.
	\end{align*}
	(note that the main difference between the two equations are in $T\log{n}$ and $T$ terms and the RHS). 
\end{lemma}
\newcommand{\colorProd}{\ensuremath{\mathcal{W}}}
\newcommand{\adversaryOracle}{\ensuremath{\mathcal{O}_A}\xspace}
\begin{proof}
	We first prove the inequality against an adaptive adversary.
	Define 
	\[
	g(n, T) := T\log{n}+n.
	\]
	Suppose, for the sake of contradiction, that there exists $T \geq n$ and an adaptive adversary oracle \adversaryOracle that can force $\partialcolor$ to have 	
	$\prob{S_T\geq g(n, T)} > n^{-T}$. 

	Using the oracle $\adversaryOracle$ to generate the update sequence, we execute $\partialcolor$ and obtain a string $R_T$ consisting of all random bits used by the algorithm up to the step where the number of calls to $\resample$ reaches $S_T=g(n, T)$, at which point we halt. We also run the compression scheme above (with some minor modifications) as described below. 

	Let $v_i$ be the vertex triggering the $i$-th call to \resample at time $t_i$. Define
	\[
		\colorProd_i := \card{\All_{t_i}(v)}, 
	\]
	representing the number of possible colorings for the neighborhood of $v_i$. Each call to \resample picks an integer from $[\colorProd_i]$ uniformly, requiring $\log \colorProd_i$ random bits. Note that $\log \colorProd_i \leq \Delta\log{\Delta}$ since
	every vertex has at most $\Delta$ neighbors and less than $\Delta$ available colors. 
	To normalize the bit-count, let $\ell_i := \Delta \log \Delta - \log \colorProd_i$. By generating $\ell_i$ additional random bits, each call uses exactly $\Delta \log \Delta$ bits. We consider these extra bits as 
	the randomness of the algorithm also. This way, the total number of bits generated by the algorithm is: 
	\[
	\abs{R_T} = \sum_{i=1}^{g(n, T)} \log \colorProd_i + \ell_i = g(n, T)\cdot \Delta\log \Delta.
	\]

	We now describe the compression. We basically do exactly as the compression scheme described earlier except that the extra $\sum \ell_i$ bits are recorded without any changes. 
	Recall that $\Bad_{t_i}(v_i)$ is the set of ``bad'' colorings for the neighborhood of $v_i$, which cause a constraint violation on this vertex. 
	By~\Cref{lem:prob-bound-bad-events}, regardless of which lines of $\fix$ is used to call $\resample$, 
	\[
		\abs{\Bad_{t_i}(v_i)} \leq \max\paren{2^{-\frac{\Delta^{\eps/2}}{100}}~,~2^{-\frac{\eps\Delta}{50\ln \Delta}}} \cdot \abs{\All_{t_i}(v_i)}\leq 2^{-\frac{\Delta^{\eps/2}}{100}} \cdot \colorProd_i;
	\]
	the last inequality is true as $\Delta$ is sufficiently large.
	The compression scheme logs 
	\[
	\log \abs{\Bad_{t_i}(v_i)} + \ell_i + O(1)
	\]
	bits per call to $\resample$ (ignoring the recursive calls inside $\resample$). Summing over all calls to $\resample$, the length of the log written for these parts is: 
		\begin{align*}
		\sum_{i=1}^{g(n, T)} \log \abs{\Bad_{t_i}(v_i)} + \ell_i + O(1)
		&\leq \sum_{i=1}^{g(n, T)} \paren{\log {\colorProd_i} - \frac{\Delta^{\eps/2}}{100} + \ell_i + O(1) } \tag{by the equation above} \\
		&= \paren{\sum_{i=1}^{g(n, T)} \log{\colorProd_i}+\ell_i + O(1)} - g(n, T) \cdot \frac{\Delta^{\eps/2}}{100}\\
		&=g(n, T)\cdot \Delta\log \Delta + g(n,T) \cdot O(1) - g(n, T) \cdot \frac{\Delta^{\eps/2}}{100}. \tag{as $\log \colorProd_i +\ell_i = \Delta\log\Delta$}
	\end{align*}
	The calculations above point to the exact step wherein we are ``compressing'' the randomness. 
	
	Finally, note that the compression scheme also stores $O(\log{\Delta})$ bits per call to $\fix$ and since the number of calls to $\fix$ and $\resample$ are asymptotically the same, 
	the log also contains $O(g(n,T) \cdot \log{\Delta})$ many bits for handling the calls to $\fix$. Denoting $L_T$ as the log generated by this compression scheme (and the modifications mentioned above), 
	we have, 
	\begin{align*}
		\card{L_T} &\leq T \cdot (2\log{n}+2) \tag{for logging the updates, their type, and if they continue to Line 3 of $\partialcolor$} \\
		&+ g(n, T)\cdot \Delta\log \Delta + g(n,T) \cdot O(\log{\Delta}) - g(n, T) \cdot \frac{\Delta^{\eps/2}}{100} \tag{for logging the calls to $\resample$ and $\fix$} \\
		&+ n\log{\Delta} \tag{to log the partial coloring upon halting the scheme after $S_T = g(n,T)$}.
	\end{align*}

	Consider the following scenario. We sample $x \in \set{0,1}^{a}$ for $a=g(n,T)\Delta\log{\Delta}$ uniformly at random and use it to generate all the randomness $R_T$ in the above process. 
	If the process fails because $S_T < g(n,T)$, we let $f(x)$ writes `fail' in $O(1)$ bits and otherwise, if $S_T = g(n,T)$, it will write `success' plus the log $L_T$ specified above. Thus, 
	by our earlier (contradicting) assumption, we will generate $L_T$ with probability at least $n^{-T}$, from which, by~\Cref{obs:compression-scheme}, we can correctly recover $x$. 
	Thus, by~\Cref{lem:compression}, we should have 
	\[
		\card{L_T} + O(1) \geq \card{R_T} - \log{(n^{T})}. 
	\]
	
	Plugging in the bounds we have for the lengths of $L_T$ and $R_T$ above, gives us
	\[
		T \cdot (2\log{n}+2) + g(n, T)\cdot \Delta\log \Delta + g(n,T) \cdot O(\log{\Delta}) - g(n, T) \cdot \frac{\Delta^{\eps/2}}{100} + n\log{\Delta} \geq g(n,T) \cdot \Delta\log{\Delta} - T\log{n},
	\]
	which implies that (since $\Delta^{\eps/2} \gg \log{\Delta}$ for large enough $\Delta$ as a function of $\eps$), 
	\[
		g(n,T) \leq \frac{200}{\Delta^{\eps/2}} \cdot \paren{4T\log{n} + n\log{\Delta}} < T\log{n}+n, 
	\]
	contradicting our definition of $g(n,T)$. Thus, we have 
	\[
		\Pr\paren{S_T > T\log{n}+n} \leq n^{-T}, 
	\]
	completing the proof for adaptive adversaries. 

	\paragraph{Improved bounds for oblivious adversaries.} For oblivious adversaries, the only difference is that we do not need to explicitly write down the $T$ updates in our log, since an oblivious adversary 
	behaves independent of the input. Define
	\[
		h(n,T) = T + n,
	\]
	and now suppose towards a contradiction that there exists a sequence of $T$ updates $e_1,\ldots,e_T$ such that running $\partialcolor$ on the graphs $G_0,\ldots,G_T$ leads to
	\[
		\Pr\paren{S_T \geq h(n,T)} \geq \Delta^{-T}. 
	\]
	If no such fixed sequence exists, then the bound in the lemma statement immediately holds. We now show the above equation will lead to a contradiction. 
	
	We again run the same exact compression scheme but now only spend $O(T)$ bits for updates just to write down whether or not $\partialcolor$ reached its Line 3, without specifying the updates since they are globally fixed. 
	Using the same exact argument as before, this allows us to summarize our random bits $R_T$ of length $h(n,T) \Delta\log{\Delta}$, with probability at least $\Delta^{-T}$, into a recoverable log $L_T$ of length 
	\[
		\card{L_T} = O(T) + h(n,T)\cdot \Delta\log{\Delta} + h(n,T) \cdot O(\log{\Delta}) - h(n, T) \cdot \frac{\Delta^{\eps/2}}{100} + n\log{\Delta}.
	\]
	Thus, by using~\Cref{lem:compression} as before, we have
	\[
		\card{L_T} \geq \card{R_T} - T\log{\Delta}, 
	\]
	implying that 
	\[
		h(n,T) \leq \frac{200}{\Delta^{\eps/2}} \cdot \paren{T\log{\Delta} + n\log{\Delta}} < T + n,
	\]
	contradicting with the choice of $h(n,T)$, and concluding the proof. 	
	\end{proof}

We are now ready to conclude the proof of \Cref{lem:partial}.

\begin{proof}[Proof of \Cref{lem:partial}]
	The algorithm $\partialcolor$ always produces a proper partial coloring. 
	This is because (1) if after an update $(u,v)$, if $\pcolort(u) = \pcolort(v)$ and neither are $\perp$, we will make one of them $\perp$, and (2) 
	the call $\resample(v)$ recolors the vertices in $N_t(v)$ and since the graph is triangle-free, $N_t(v)$ is an independent set, and thus each vertex $u \in N_t(v)$ can 
	independently choose a color uniformly at random from $\availt{u}\ \cup \set{\perp}$ without creating a monochromatic edge. Consequently, the resulting coloring remains proper. 
	
	By induction, suppose $G_{t-1}$ satisfies the properties stated in \Cref{lem:partial}. Our goal is to show that after processing time step $t$, the algorithm $\partialcolor$ restores all constraints. The only difference between $G_t$ and $G_{t-1}$ arises from the edge $e_t$. If the corresponding constraints are violated, the subroutine $\fix$ is invoked. By \Cref{obs:flaw-resolved}, all violated constraints are resolved by the call to $\fix$, and no new violations are introduced. Therefore, after running $\partialcolor$, for all $v \in V$ we have
\[
	a_t(v) \geq \Delta^{\eps/2}
	\qquad \text{and} \qquad
	f_t(v) \leq \frac{\eps q}{2}.
\]

	For any $T \geq n$, when facing an adaptive adversary, by \Cref{lem:upper-bound-resample} we know that with probability at least $1-n^{-T}$, 
	the subroutine \resample is invoked at most $O(T\log{n})$ time. Each call to $\resample(v)$ takes $O(\Delta^2)$ times (for sampling colors for $O(\Delta)$ vertices and going over $B_t(v,2)$ of size $\Delta^2$) and each call
	to $\fix(v)$ takes $O(\Delta)$ time (to go over all neighbors of $v$ in Line 2 if needed). Since the number of calls to $\resample$ and $\fix$ is asymptotically the same, the total runtime of the algorithm
	after $T$ updates is $O(T \cdot \Delta^2\log{n})$ with probability $1-n^{-T}$. In other words, the amortized update time of the algorithm is $O(\Delta^2\log{n})$ with probability at least $1-n^{-T}$. 
	Taking union bound for $T$ from $n$ to infinity, we have, 
	\[
		\Pr\paren{\text{amortized update time is  ever more than $O(\Delta^2\log{n})$}} \leq \sum_{T=n}^{\infty} n^{-T} = O(n^{-n}), 
	\]
	namely, with (super exponentially) high probability, the amortized update time for any sequence of updates of length at least $n$ is $O(\Delta^2\log{n})$. The bound for the recourse can also be proven similarly to be $O(\Delta\log{n})$ since each call to $\resample(v)$ has a recourse of $O(\Delta)$ at most. 
	
	Finally, switching to oblivious adversaries, using the improved bounds of \Cref{lem:upper-bound-resample}, the same argument instead gives us
	\[
		\Pr\paren{\text{amortized update time is  ever more than $O(\Delta^2)$}} \leq \sum_{T=n}^{\infty} \Delta^{-T} = O(\Delta^{-n}),
	\]
	which again implies the desired bounds with (exponentially) high probability. 
\end{proof}

\subsection{Concluding the Proof of~\Cref{thm:poly-Delta}}
We are now ready to conclude the proof of~\Cref{thm:poly-Delta}. For this, we present the following algorithm. 
Upon each update, we first invoke $\partialcolor$ to resolve any conflicts among vertices colored from $\CC_1$. Then, any vertex that remains
uncolored is assigned a color from $\CC_2$ using a simple greedy procedure. Note that
\[
\abs{\CC_1 \cup \CC_2} = (1+2\eps)\cdot \frac{\Delta}{\ln\Delta}.
\]
By initially replacing $\eps$ with $\eps/2$, this yields a $(1+\eps)\Delta/\ln \Delta$ coloring.

\begin{Algorithm}\label{alg:dynamic-main}
	Called when $e_t=(u_t, v_t)$ is updated in $G_{t-1}$ to obtain $G_t$.
	\begin{enumerate}
		\item Call $\partialcolor(e_t)$ and store all vertices whose color changes in a set $L$.
		\item For each $v \in L$ with $\pcolort(v) = \perp$, assign $v$ an available color by checking the neighborhood of $v$ and finding a color in $\CC_2$ not assigned to any of its neighbors.
	\end{enumerate}
\end{Algorithm}

It is worth noting that when running $\partialcolor$, any vertex $u$ that previously received a color from $\CC_2$ can be treated as having color $\perp$ (this way, the coloring $\pcolort$ of $\partialcolor$ does not look at 
colors of vertices in $\CC_2$, nor any vertex colored from $\CC_2$ in Line 2 of~\Cref{alg:dynamic-main} ever changes $\pcolort$). 
We now use \Cref{alg:dynamic-main} to prove \Cref{thm:poly-Delta}.

\begin{proof}[Proof of \Cref{thm:poly-Delta}]

	At each update step $t$, the subroutine $\partialcolor$ maintains a proper partial coloring using the color set $\CC_1$, and for every vertex $v$ we have $f_t(v) \leq \frac{\eps q}{2}$. Since $\abs{\CC_2} = \eps q$, there are sufficiently many available colors to greedily color all remaining uncolored vertices using $\CC_2$ in the second phase of \Cref{alg:dynamic-main}. So \Cref{alg:dynamic-main} maintains a proper $(1+\eps)\frac{\Delta}{\ln \Delta}$ coloring (after rescaling $\eps \leftarrow \eps/2$).

	The runtime of the first step of~\Cref{alg:dynamic-main} is exactly the same as that of $\partialcolor$. For the second step, the runtime is $O(\Delta)$ time per vertex in $L$. Recall that vertices in $L$ are 
	the ones whose color changed by $\partialcolor$ in this step and thus are bounded by the recourse of $\partialcolor$. This implies that the amortized update time of the algorithm, against an adaptive adversary, is 
	\[
		O(\Delta^2\log{n}) + O(\Delta\log{n}) \cdot O(\Delta) = O(\Delta^2\log{n}),
	\]
	with high probability by~\Cref{lem:partial}. The bound of $O(\Delta^2)$ on the amortized update time follows exactly the same using the improved bound of~\Cref{lem:partial} for oblivious adversaries. 
\end{proof}


\newcommand{\algA}[1]{\ensuremath{\mathbb{A}_{#1}}}
\newcommand{\partition}[2]{\ensuremath{V_{#1}^{(#2)}}}
\newcommand{\indeg}[1]{\ensuremath{\textnormal{indeg}(#1)}}
\newcommand{\colorPart}[1]{\ensuremath{\mathcal{C}^{(#1)}}}

\section{A $\Delta^{o(1)}\log{n}$ Update Time Algorithm}\label{sec:Delta-reduction}  

We now present our final algorithm  for maintaining an $O(\Delta/\ln\Delta)$ coloring of triangle-free graphs with $\Delta^{o(1)}\log{n}$ amortized update time. The following theorem formalizes \Cref{res:main}.

\begin{theorem}[Formalizing~\Cref{res:main}]\label{thm:Delta-o(1)}
	For any fixed $\gamma \in (0,1)$, there exists an integer $K = O(1/\gamma^2)$ such that the following holds. 
	There is a randomized algorithm that with high probability maintains a $\paren{K \cdot {\Delta}/{\ln \Delta}}$ coloring of any triangle-free fully dynamic graph with maximum degree $\Delta$ with amortized update time of $O(\Delta^{\gamma} \cdot \log{n})$ 
	against an adaptive adversary and $O(\Delta^\gamma)$ against an oblivious adversary. 
\end{theorem}

To prove this theorem, we design a family of recursive algorithms $\{\algA{k}\}_{k=1}^{\infty}$ with progressively better update times. The base case of this family, $\algA{1}$, is \Cref{alg:dynamic-main} which has update time $\approx \Delta^2$. 
For larger values of $k > 1$, $\algA{k}$ achieves $\approx \Delta^{O(1/k)}$ amortized update time.

\begin{lemma}\label{lem:recursive-alg}
For any integer $k\geq 1$, there is a randomized algorithm that with high probability maintains an
$
		O(k^2 \cdot {\Delta}/{\ln{\Delta}})
	$
	coloring of any fully dynamic $n$-vertex triangle-free graph of maximum degree $\Delta$ with amortized update time: 
	\begin{itemize}
		\item $O\paren{k! \cdot \Delta^{\frac{2}{2k-1}}\cdot\log{n}}$ against an adaptive adversary; and, 
		\item $O\paren{k! \cdot \Delta^{\frac{2}{2k-1}}}$ against an oblivious adversary. 
	\end{itemize}
\end{lemma}

\Cref{thm:Delta-o(1)} follows from~\Cref{lem:recursive-alg} immediately by setting $k=\Theta(1/\gamma)$ and $K = O(k^2)$ (and noting that $\gamma$ is fixed independent of $n$ and thus the $k!$ term in the update time is suppressed  
in the asymptotic notation.) Note that in the limit, one can set $k$ as large as some $k = o(\sqrt{\log{\Delta}})$ to maintain an $o(\Delta)$ coloring of triangle-free graphs in ${\Delta^{\Theta(\frac{\log\log{\Delta}}{\sqrt{\log{\Delta}}})}\log{n}}$ amortized update 
time. We prove~\Cref{lem:recursive-alg} in the rest of this section.

For any integer $k \geq 1$, algorithm $\algA{k}$ receives an upper bound $\Delta$ on the maximum degree of the graph it will ever encounter, and the palette $\CC$ of the available colors in the \textbf{initialization phase}. 
Then, during the \textbf{update phase}, $\algA{k}$ receives updates of the form $(u,v,\pm)$, for $u,v \in V$, which corresponds to inserting or deleting the edge $(u,v)$ from the underlying graph.  

The algorithm $\algA{1}$ is our algorithm in~\Cref{thm:poly-Delta}. We now describe $\algA{k}$ for any $k > 1$ on an input dynamic graph $G$ on vertices $V$ with maximum degree $\Delta$. Define the following two parameters 
\begin{align}
	b_k = b_k(\Delta):= \frac{1}{2} \cdot \Delta^{\frac{2}{2k-1}} \qquad \text{and} \qquad c_k = c_k(\Delta) := (k+1) \cdot (2k-1) \cdot\frac{\Delta}{\ln \Delta}. \label{eq:para-o(1)}
\end{align}
We require the palette $\CC$ of available colors to $\algA{k}$ to be of size $c_k$. This palette is further partitioned arbitrarily into $b_k$ sets $\CC_1,\cdots,\CC_{b_k}$ of equal size. 

For every vertex $v \in V$, the algorithm maintains the following data
structures that is updated in each time step:
\begin{itemize}
    \item A partition of the vertex set into $b_k$ sets $V_1,\ldots,V_{b_k}$. For every vertex $v \in V$, we use $p(v) \in [b_k]$ to denote the index of the partition it belongs to. Initially, vertices are  partitioned arbitrarily. 
      \item $\indeg{v}$: the degree of $v$ in the \emph{induced} subgraph of $G$ on $V_{p(v)}$, namely, the partition $v$ belongs to. We refer to $\indeg{v}$ as the \textbf{induced degree} of $v$. 
    \item $N(v,i)$: the neighbors of $v$ inside the induced subgraph on $V_i$, denoted by $G[V_i]$. 
\end{itemize}
The goal of the algorithm is to color vertices in each set $V_i$ using only the colors in $\CC_i$. To ensure possibility of this, we maintain an invariant on the maximum degree of each $G[V_i]$ for $i \in [b_k]$. 

\begin{invariant}\label{inv:degree}
For any vertex $v \in V$, $\indeg{v} \leq \paren{1+\frac{1}{k}} \cdot \frac{\Delta}{b_k}$. 
\end{invariant}

Equipped with~\Cref{inv:degree}, the algorithm $\algA{k}$ runs $b_k$ copies of the algorithm $\algA{k-1}$, one on each subgraph $G[V_i]$ and color palette $\CC_i$ for $i \in [b_k]$ with an upper bound of $\Delta(G[V_i]) \leq (1+1/k)\Delta/b_k$ on the maximum degree of the subgraph. Since the palettes $\CC_i$'s are disjoint, this will lead to a proper coloring of the entire graph. We can now formalize the algorithm. 

\newcommand{\algAcopy}[2]{\ensuremath{\mathcal{A}_{#1}^{(#2)}}}
\newcommand{\fixUpdate}{\ensuremath{\textnormal{\textsc{FixUpdate}}}\xspace}
\newcommand{\fixInv}{\ensuremath{\textnormal{\textsc{FixInvariant}}}\xspace}
\newcommand{\Time}[1]{\ensuremath{T_{#1}}}

\begin{Algorithm}[Algorithm $\algA{k}$ of \Cref{lem:recursive-alg}]\label{alg:algk}
    
    ~
    \vspace{-5pt}
    \begin{enumerate}
    	\item \textbf{Initialization phase:} Initialize $b_k$ copies of $\algA{k-1}$, where the $i$-th copy, denoted by $\algA{k-1}[i]$ is given the palette $\CC_i$ and the promised upper bound on maximum degree $(1+1/k)\Delta/b_k$. 
	During the updates, $\algA{k-1}[i]$ will be given the dynamic graph $G_i = (V , E[V_i])$, the graph with vertices $V$, but only edges 
	that are between vertices of the set $V_i$ of the partition\footnote{Setting vertices of $G_i$ to be $V$ instead of $V_i$ is simply to bypass having to define operations for inserting or deleting vertices in addition to edges.}.    
	
	\item \textbf{Update phase:} Given an update $(u_t,v_t,\pm)$ to $G$ at time $t$: 
	\begin{enumerate}
		\item Update  $N(u_t, p(v_t)),N(v_t, p(u_t))$ and $\indeg{u_t},\indeg{v_t}$ accordingly. 
		\item If~\Cref{inv:degree} no longer holds, run the subroutine $\fixInv$. 
		\item Once~\Cref{inv:degree} continues to hold, if $p(u_t) \neq p(v_t)$, terminate, otherwise pass the update $(u_t,v_t,\pm)$ to $\algA{k-1}[i]$, where $i=p(u_t)=p(v_t)$,
		 to update the colors. 
	\end{enumerate}
   \end{enumerate}
   \noindent
    \underline{\textbf{Subroutine} $\fixInv$:}\label{subroutine:fix-update}
    \begin{enumerate}
    	\item Let $w$ be any vertex violating the invariant and let $i = p(w)$. 
	\item Remove all edges of $w$ in $G_i$ (i.e., inside $E[V_i]$) by calling $\algA{k-1}[i]$ for each edge. 
	\item Find an index $j \in [b_{k+1}]$ such that $\card{N(w, j)} \leq \Delta/b_k$ (such an index exist by pigeonhole principle). Remove $w$ from $V_i$ and insert it to $V_j$ and update $p(w)$. 
	\item Insert all edges of $w$ incident on $E[V_j]$ to the graph $G_j$ by calling $\algA{k-1}[j]$ for each edge. 
	\item If~\Cref{inv:degree} still does not hold, call $\fixInv$ again. 
    \end{enumerate}
\end{Algorithm}

We use a simple potential function argument to argue that the number of calls to $\fixInv$ (including recursive calls inside it) across a fixed number of updates will be bounded. We will use this later both in the proof of the correctness of the algorithm
and to bound its runtime. 

\begin{lemma}\label{lem:bounded-fix-inv}
	Starting from an empty graph $G$ as input to $\algA{k}$, over any sequence of $T$ updates (even by an adaptive adversary), the total number of times $\fixInv$ is called is $O(T \cdot k \cdot b_k/\Delta)$.  
\end{lemma}
\begin{proof}
	Define the potential function: 
	\[
		\Phi := \sum_{v \in V} \indeg{v}. 
	\]
	Each update can increase $\Phi$ by at most two. On the other hand, whenever $\fixInv$ is called, one vertex $w$ with $\indeg{w} > (1+1/k) \cdot \Delta/b_k$ moves from some set $V_i$ to another $V_j$ resulting in $\indeg{w} \leq \Delta/b_k$. 
	Removal of $w$ from $V_i$ reduces $\indeg{w}$ and induced degrees of neighbors in $V_i$ by at least $2(1+1/k)\Delta/b_k$ in total. On the other hand, including $w$ in $V_j$ increases the induced degrees of vertices 
	by at most $2\Delta/b_k$. As such, this move reduces $\Phi$ by at least $2/k \cdot \Delta/b_k$. 
	
	Thus, over a sequence of $T$ updates, we increase $\Phi$ to at most $2T$ and $\Phi$ can never become negative, thus there can be at most 
	\[
		\frac{2T}{2/k \cdot \Delta/b_k} = \frac{T \cdot k \cdot b_k}{\Delta},
	\]
	many calls to $\fixInv$, as desired.  
\end{proof}

\Cref{lem:bounded-fix-inv} ensures that the subroutine $\fixInv$ terminates in $\algA{k}$ and thus we can safely assume~\Cref{inv:degree} when running the algorithm. Using this, we can prove the correctness of the 
algorithm. Before that, we should note while the implementation of $\algA{k}$ for $k > 1$ in~\Cref{alg:algk} is deterministic, the base case algorithm, namely, $\algA{1}$ from~\Cref{thm:poly-Delta}, is a randomized algorithm. Thus, each of $\algA{k}$ 
is also a randomized algorithm.

\begin{lemma}\label{lem:coloring-invariant}
With high probability, $\algA{k}$ computes a proper coloring of the input graph with $c_k$ colors. 
\end{lemma}
\begin{proof}
	Firstly, consider the dynamic graph $G_i=(V,E[V_i])$ defined throughout the updates. This graph is updated through direct adversary updates to $G$ plus updates done in $\fixInv$. 
	Here we should note that to accurately implement $\fixInv$, we should technically first collect
	all insertion and deletion of edges decided across all calls of $\fixInv$ inside a single original update to the graph, and then, once $\fixInv$ is entirely finished do all the insertions and deletions at the end. This way, 
	(1) all updates to the dynamic graph $G_i$ are passed to it via $\algA{k-1}[i]$, and (2) we can assume an upper bound of $(1+1/k)\Delta/b_k$ on the maximum degree of graph $G_i$. This means all calls to $\algA{k-1}[i]$ 
	for $i \in [b_k]$ are valid throughout the execution of $\algA{k}$.

	We know that~\Cref{inv:degree} holds and for any $i \in [k]$, $\algA{k-1}[i]$ is run entirely on the dynamic graph $G_i = (V,E[V_i])$ with maximum degree $(1+1/k)\Delta/b_k$. 
	By~\Cref{eq:para-o(1)}, this means that $\algA{k-1}[i]$ requires a palette of size 
	\begin{align*}
		c_{k-1}(\frac{(1+1/k) \cdot \Delta}{b_k}) &= k \cdot (2k-3) \cdot \frac{(1+1/k) \cdot \Delta}{b_k \cdot \ln{((1+1/k) \cdot \Delta/b_k)}} \tag{by~\Cref{eq:para-o(1)} for $c_{k-1}$} \\
		&=  (k+1) \cdot (2k-3) \cdot \frac{\Delta}{b_k \cdot \ln{(2 \cdot (1+1/k) \cdot \Delta^{\frac{2k-3}{2k-1}}})} \tag{by~\Cref{eq:para-o(1)} for $b_k$} \\
		&\leq (k+1) \cdot (2k-3) \cdot \frac{2k-1}{2k-3}	\cdot \frac{\Delta}{b_k \cdot \ln{\Delta}} \\
		&= \frac{c_k}{b_k} \tag{by~\Cref{eq:para-o(1)} for $c_{k}$}. 
	\end{align*}
	This is precisely the size of the palette $\CC_i$ provided to $\algA{k-1}[i]$. Thus, as long as $\algA{k-1}[i]$, for each $i \in [b_k]$, works correctly, so does $\algA{k}$. 
	Finally, since $V_1,\ldots,V_{b_k}$ is a partition of $V$ and we are using disjoint palettes for each $V_i$, the union of colorings found by $\algA{k-1}[i]$'s is also a proper
	coloring of the entire graph. This concludes the proof. 
\end{proof}

Finally, we can analyze the update time of the algorithm. In the following, we focus on adaptive adversaries first and then point out how to improve the runtime for oblivious adversaries. For any $k \geq 1$, define: 
\begin{center}
	$\Time{k}(\Delta):$ a high probability upper bound on the amortized update time of $\algA{k}$ \\ on a graph of maximum degree $\Delta$ against an adaptive adversary. 
\end{center}

We prove the following upper bound on $\Time{k}(\cdot)$. 

\begin{lemma}\label{lem:time-k}
	For any $k \geq 1$, 
	\[
		\Time{k}(\Delta) = O(k! \cdot \Delta^{\frac{2}{2k-1}}\cdot\log{n}). 
	\]
\end{lemma}
\begin{proof}
	By~\Cref{thm:poly-Delta}, for $\algA{1}$, we have 
	\[
		\Time{1}(\Delta) = O(\Delta^2\log{n}), 
	\]
	which fits the desired bounds for $k=1$. This establishes the base case. We now prove the bounds for any $k > 1$ by analyzing~\Cref{alg:algk}. 
	
	The amortized update time of~\Cref{alg:algk}, ignoring the time to call $\fixInv$ for now, is $O(1)$ for updating the data structures, 
	plus running one copy of $\algA{k-1}$ on a graph of maximum degree $O(\Delta/b_k)$. By induction, and the choice of $b_k$ in~\Cref{eq:para-o(1)}, this takes 
	\[
		\Time{k-1}\paren{O(\frac{\Delta}{b_k})} = \Time{k-1}\paren{O(\Delta^{\frac{2k-3}{2k-1}})} = O((k-1)! \cdot \Delta^{\frac{2}{2k-1}}\cdot\log{n}). 
	\]
	
	We now bound the runtime of the calls to $\fixInv$. Firstly, in each call, it takes $O(\Delta)$ time to move the considered vertex $w$ from one set $V_i$ to another set $V_j$. 
	Moreover, it takes another $O(\Delta/b_k)$ many calls to $\algA{k-1}$ to remove and insert edges of $w$. 	Finally, to implement $\fixInv$ efficiently, we can keep a list of vertices whose degree have changed in each call to $\fixInv$ (when we update the edge
	incident on that vertex) so that we can find a vertex $w$ that violates \Cref{inv:degree} (if any) in $O(1)$ time without having 
	to search the entire graph. 
	
	Thus, a single call to $\fixInv$ takes 
	\[
		O(\Delta) + O(\Delta/b_k) \cdot \Time{k-1}\paren{O(\frac{\Delta}{b_k})}
	\]
	time in total. Moreover, by~\Cref{lem:bounded-fix-inv}, over a sequence of $T$ updates, the total number of calls to $\fixInv$ is $O(k \cdot T \cdot b_k/\Delta)$. Thus, the amortized update time of this step will be 
	\begin{align*}
		O(k \cdot b_k) + O(k) \cdot \Time{k-1}\paren{O(\frac{\Delta}{b_k})} &= O(k \cdot \Delta^{\frac{2}{2k-1}}) + O(k) \cdot O((k-1)! \cdot \Delta^{\frac{2}{2k-1}}\cdot\log{n}) \\
		&= O(k! \cdot \Delta^{\frac{2}{2k-1}}\cdot\log{n}), 
	\end{align*}
	concluding the proof. 
\end{proof}

\Cref{lem:coloring-invariant,lem:time-k} together conclude the proof of~\Cref{lem:recursive-alg} for the adaptive strategies. We should note that since over a sequence of $T$ updates, 
we are running $T \cdot \poly(n)$ many copies of $\algA{1}$ whereas error probability of $\algA{1}$ is only $n^{-T}$, we can do a union bound over all probabilistic guarantees required in our proof. 

Finally, to prove~\Cref{lem:recursive-alg} for oblivious adversaries, notice that by~\Cref{thm:poly-Delta}, against an oblivious adversary, $\algA{1}$ has an amortized update time of $O(\Delta^2)$ instead. Thus, the same proof of~\Cref{lem:time-k}, 
with the base case of $O(\Delta^2\log{n})$ switched with $O(\Delta^2)$, implies that the amortized update time of $\algA{k}$ has no dependence on $\log{n}$ and is $O(k! \cdot \Delta^{\frac{2}{2k-1}})$ time instead. 

This fully concludes the proof of~\Cref{lem:recursive-alg} and by extension~\Cref{thm:Delta-o(1)}.

\section*{Acknowledgement} 
\addcontentsline{toc}{section}{Acknowledgement}

We are thankful to Thatchaphol Saranurak for illuminating discussions and pointing us to the proactive resampling technique in~\cite{BernsteinBGNSS022,BhattacharyaSS22} and Bernhard Haeupler for discussions about his and collaborators recent work in~\cite{HaeuplerP26,HaeuplerMRST26} and coordinating 
our arXiv submissions. 

\bibliographystyle{alpha}
\bibliography{general}

@article{Reed98,
  title={{$\omega$}, {$\Delta$}, and {$\chi$}},
  author={Reed, Bruce},
  journal={Journal of Graph Theory},
  volume={27},
  number={4},
  pages={177--212},
  year={1998},
  publisher={Wiley Online Library}
}

@inproceedings{AchlioptasIS19,
  author       = {Dimitris Achlioptas and
                  Fotis Iliopoulos and
                  Alistair Sinclair},
  editor       = {David Zuckerman},
  title        = {Beyond the Lov{\'{a}}sz Local Lemma: Point to Set Correlations
                  and Their Algorithmic Applications},
  booktitle    = {60th {IEEE} Annual Symposium on Foundations of Computer Science, {FOCS}
                  2019, Baltimore, Maryland, USA, November 9-12, 2019},
  pages        = {725--744},
  publisher    = {{IEEE} Computer Society},
  year         = {2019}
  }

@article{Bernshteyn19,
  title={The Johansson-Molloy theorem for DP-coloring},
  author={Bernshteyn, Anton},
  journal={Random Structures \& Algorithms},
  volume={54},
  number={4},
  pages={653--664},
  year={2019},
  publisher={Wiley Online Library}
}

@inproceedings{BeimelKMNSS22,
  author       = {Amos Beimel and
                  Haim Kaplan and
                  Yishay Mansour and
                  Kobbi Nissim and
                  Thatchaphol Saranurak and
                  Uri Stemmer},
  editor       = {Stefano Leonardi and
                  Anupam Gupta},
  title        = {Dynamic algorithms against an adaptive adversary: generic constructions
                  and lower bounds},
  booktitle    = {{STOC} '22: 54th Annual {ACM} {SIGACT} Symposium on Theory of Computing,
                  Rome, Italy, June 20 - 24, 2022},
  pages        = {1671--1684},
  publisher    = {{ACM}},
  year         = {2022}
  }

@inproceedings{BernsteinBKS25,
  author       = {Aaron Bernstein and
                  Sayan Bhattacharya and
                  Peter Kiss and
                  Thatchaphol Saranurak},
  editor       = {Michal Kouck{\'{y}} and
                  Nikhil Bansal},
  title        = {Deterministic Dynamic Maximal Matching in Sublinear Update Time},
  booktitle    = {Proceedings of the 57th Annual {ACM} Symposium on Theory of Computing,
                  {STOC} 2025, Prague, Czechia, June 23-27, 2025},
  pages        = {132--143},
  publisher    = {{ACM}},
  year         = {2025}
  }

@inproceedings{AlonA20,
  author       = {Noga Alon and
                  Sepehr Assadi},
  editor       = {Jaroslaw Byrka and
                  Raghu Meka},
  title        = {Palette Sparsification Beyond ({\(\Delta\)}+1) Vertex Coloring},
  booktitle    = {Approximation, Randomization, and Combinatorial Optimization. Algorithms
                  and Techniques, {APPROX/RANDOM} 2020, Virtual Conference, August 17-19,
                  2020},
  series       = {LIPIcs},
  volume       = {176},
  pages        = {6:1--6:22},
  publisher    = {Schloss Dagstuhl - Leibniz-Zentrum f{\"{u}}r Informatik},
  year         = {2020}
  }

@article{BonamyKNP22,
  title={Bounding {$\chi$} by a fraction of {$\Delta$} for graphs without large cliques},
  author={Bonamy, Marthe and Kelly, Tom and Nelson, Peter and Postle, Luke},
  journal={Journal of Combinatorial Theory, Series B},
  volume={157},
  pages={263--282},
  year={2022},
  publisher={Elsevier}
}

@inproceedings{PettieS13,
  author       = {Seth Pettie and
                  Hsin{-}Hao Su},
  editor       = {Fedor V. Fomin and
                  Rusins Freivalds and
                  Marta Z. Kwiatkowska and
                  David Peleg},
  title        = {Fast Distributed Coloring Algorithms for Triangle-Free Graphs},
  booktitle    = {Automata, Languages, and Programming - 40th International Colloquium,
                  {ICALP} 2013, Riga, Latvia, July 8-12, 2013, Proceedings, Part {II}},
  series       = {Lecture Notes in Computer Science},
  volume       = {7966},
  pages        = {681--693},
  publisher    = {Springer},
  year         = {2013}
  }

@article{Jamall11,
  title={A coloring algorithm for triangle-free graphs},
  author={Jamall, Mohammad Shoaib},
  journal={arXiv preprint arXiv:1101.5721},
  year={2011}
}

@inproceedings{FlinH25,
  author       = {Maxime Flin and
                  Magn{\'{u}}s M. Halld{\'{o}}rsson},
  editor       = {Keren Censor{-}Hillel and
                  Fabrizio Grandoni and
                  Jo{\"{e}}l Ouaknine and
                  Gabriele Puppis},
  title        = {Faster Dynamic {$(\Delta+1)$}-Coloring Against Adaptive Adversaries},
  booktitle    = {52nd International Colloquium on Automata, Languages, and Programming,
                  {ICALP} 2025, Aarhus, Denmark, July 8-11, 2025},
  series       = {LIPIcs},
  volume       = {334},
  pages        = {79:1--79:21},
  publisher    = {Schloss Dagstuhl - Leibniz-Zentrum f{\"{u}}r Informatik},
  year         = {2025}
  }

@inproceedings{ChristiansenR22,
  author       = {Aleksander B. G. Christiansen and
                  Eva Rotenberg},
  editor       = {Mikolaj Bojanczyk and
                  Emanuela Merelli and
                  David P. Woodruff},
  title        = {Fully-Dynamic {$\alpha+2$} Arboricity Decompositions and Implicit
                  Colouring},
  booktitle    = {49th International Colloquium on Automata, Languages, and Programming,
                  {ICALP} 2022, Paris, France, July 4-8, 2022},
  series       = {LIPIcs},
  volume       = {229},
  pages        = {42:1--42:20},
  publisher    = {Schloss Dagstuhl - Leibniz-Zentrum f{\"{u}}r Informatik},
  year         = {2022}
  }

@inproceedings{ChristiansenNR23,
  author       = {Aleksander Bj{\o}rn Grodt Christiansen and
                  Krzysztof Nowicki and
                  Eva Rotenberg},
  editor       = {Barna Saha and
                  Rocco A. Servedio},
  title        = {Improved Dynamic Colouring of Sparse Graphs},
  booktitle    = {Proceedings of the 55th Annual {ACM} Symposium on Theory of Computing,
                  {STOC} 2023, Orlando, FL, USA, June 20-23, 2023},
  pages        = {1201--1214},
  publisher    = {{ACM}},
  year         = {2023}
  }

@article{HenzingerP22,
  author       = {Monika Henzinger and
                  Pan Peng},
  title        = {Constant-time Dynamic {$(\Delta+1)$}-Coloring},
  journal      = {{ACM} Trans. Algorithms},
  volume       = {18},
  number       = {2},
  pages        = {16:1--16:21},
  year         = {2022}
  }

@article{BhattacharyaGKL22,
  author       = {Sayan Bhattacharya and
                  Fabrizio Grandoni and
                  Janardhan Kulkarni and
                  Quanquan C. Liu and
                  Shay Solomon},
  title        = {Fully Dynamic {$(\Delta+1)$}-Coloring in {$O(1)$} Update Time},
  journal      = {{ACM} Trans. Algorithms},
  volume       = {18},
  number       = {2},
  pages        = {10:1--10:25},
  year         = {2022}
  }

@inproceedings{ChristiansenRV24,
  author       = {Aleksander B. G. Christiansen and
                  Eva Rotenberg and
                  Juliette Vlieghe},
  editor       = {Hans L. Bodlaender},
  title        = {Sparsity-Parameterised Dynamic Edge Colouring},
  booktitle    = {19th Scandinavian Symposium and Workshops on Algorithm Theory, {SWAT}
                  2024, Helsinki, Finland, June 12-14, 2024},
  series       = {LIPIcs},
  volume       = {294},
  pages        = {20:1--20:18},
  publisher    = {Schloss Dagstuhl - Leibniz-Zentrum f{\"{u}}r Informatik},
  year         = {2024}
  }

@inproceedings{DuanHZ19,
  author       = {Ran Duan and
                  Haoqing He and
                  Tianyi Zhang},
  editor       = {Timothy M. Chan},
  title        = {Dynamic Edge Coloring with Improved Approximation},
  booktitle    = {Proceedings of the Thirtieth Annual {ACM-SIAM} Symposium on Discrete
                  Algorithms, {SODA} 2019, San Diego, California, USA, January 6-9,
                  2019},
  pages        = {1937--1945},
  publisher    = {{SIAM}},
  year         = {2019}
  }

@inproceedings{BehnezhadRW25,
  author       = {Soheil Behnezhad and
                  Rajmohan Rajaraman and
                  Omer Wasim},
  editor       = {Yossi Azar and
                  Debmalya Panigrahi},
  title        = {Fully Dynamic {$(\Delta+1)$}-Coloring Against Adaptive Adversaries},
  booktitle    = {Proceedings of the 2025 Annual {ACM-SIAM} Symposium on Discrete Algorithms,
                  {SODA} 2025, New Orleans, LA, USA, January 12-15, 2025},
  pages        = {4983--5026},
  publisher    = {{SIAM}},
  year         = {2025}
  }

@inproceedings{BhattacharyaCPS24,
  author       = {Sayan Bhattacharya and
                  Mart{\'{\i}}n Costa and
                  Nadav Panski and
                  Shay Solomon},
  editor       = {David P. Woodruff},
  title        = {Nibbling at Long Cycles: Dynamic (and Static) Edge Coloring in Optimal
                  Time},
  booktitle    = {Proceedings of the 2024 {ACM-SIAM} Symposium on Discrete Algorithms,
                  {SODA} 2024, Alexandria, VA, USA, January 7-10, 2024},
  pages        = {3393--3440},
  publisher    = {{SIAM}},
  year         = {2024}
  }

@inproceedings{BhattacharyaCHN18,
  author       = {Sayan Bhattacharya and
                  Deeparnab Chakrabarty and
                  Monika Henzinger and
                  Danupon Nanongkai},
  editor       = {Artur Czumaj},
  title        = {Dynamic Algorithms for Graph Coloring},
  booktitle    = {Proceedings of the Twenty-Ninth Annual {ACM-SIAM} Symposium on Discrete
                  Algorithms, {SODA} 2018, New Orleans, LA, USA, January 7-10, 2018},
  pages        = {1--20},
  publisher    = {{SIAM}},
  year         = {2018}
  }

@book{MolloyR02,
  title={\textbf{Graph colouring and the probabilistic method}},
  author={Molloy, Michael and Reed, Bruce},
  volume={23},
  year={2002},
  publisher={Springer Science \& Business Media}
}

@inproceedings{BansalGG15,
  author       = {Nikhil Bansal and
                  Anupam Gupta and
                  Guru Guruganesh},
  editor       = {Rocco A. Servedio and
                  Ronitt Rubinfeld},
  title        = {On the Lov{\'{a}}sz Theta function for Independent Sets in Sparse
                  Graphs},
  booktitle    = {Proceedings of the Forty-Seventh Annual {ACM} on Symposium on Theory
                  of Computing, {STOC} 2015, Portland, OR, USA, June 14-17, 2015},
  pages        = {193--200},
  publisher    = {{ACM}},
  year         = {2015}
  }

@article{AlonKS99,
  title={Coloring graphs with sparse neighborhoods},
  author={Alon, Noga and Krivelevich, Michael and Sudakov, Benny},
  journal={Journal of Combinatorial Theory, Series B},
  volume={77},
  number={1},
  pages={73--82},
  year={1999},
  publisher={Elsevier}
}

@article{Reed99,
  title={A strengthening of Brooks' theorem},
  author={Reed, Bruce},
  journal={Journal of Combinatorial Theory, Series B},
  volume={76},
  number={2},
  pages={136--149},
  year={1999},
  publisher={Academic Press}
}

@article{MolloyR14,
  title={Colouring graphs when the number of colours is almost the maximum degree},
  author={Molloy, Michael and Reed, Bruce},
  journal={Journal of Combinatorial Theory, Series B},
  volume={109},
  pages={134--195},
  year={2014},
  publisher={Elsevier}
}

@techreport{Johansson96,
  title={Asymptotic choice number for triangle free graphs},
  author={Johansson, Anders},
  year={1996},
  institution={Technical report 91-5, DIMACS}
}

@article{Kim95,
  title={On Brooks' theorem for sparse graphs},
  author={Kim, Jeong Han},
  journal={Combinatorics, Probability and Computing},
  volume={4},
  number={2},
  pages={97--132},
  year={1995},
  publisher={Cambridge University Press}
}

@inproceedings{Brooks41,
  title={On colouring the nodes of a network},
  author={Brooks, Rowland Leonard},
  booktitle={Mathematical Proceedings of the Cambridge Philosophical Society},
  volume={37},
  number={2},
  pages={194--197},
  year={1941},
  organization={Cambridge University Press}
}

@article{PS97,
author = {Panconesi, Alessandro and Srinivasan, Aravind},
title = {Randomized Distributed Edge Coloring via an Extension of the Chernoff--Hoeffding Bounds},
year = {1997},
issue_date = {April 1997},
publisher = {Society for Industrial and Applied Mathematics},
address = {USA},
volume = {26},
number = {2},
issn = {0097-5397},
url = {https://doi.org/10.1137/S0097539793250767},
doi = {10.1137/S0097539793250767},
journal = {SIAM J. Comput.},
month = apr,
pages = {350–368},
numpages = {19}
}

@misc{Molloy18,
      title={The list chromatic number of graphs with small clique number}, 
      author={Michael Molloy},
      year={2018},
      eprint={1701.09133},
      archivePrefix={arXiv},
      primaryClass={math.CO},
      url={https://arxiv.org/abs/1701.09133}, 
}

@article{joagProschan1983,
  title={Negative association of random variables with applications},
  author={Joag-Dev, Kumar and Proschan, Frank},
  journal={The Annals of Statistics},
  pages={286--295},
  year={1983},
  publisher={JSTOR}
}

@inproceedings{GrableP98,
  author       = {David A. Grable and
                  Alessandro Panconesi},
  editor       = {Howard J. Karloff},
  title        = {Fast Distributed Algorithms for \{Brooks-Vizing\} Colourings},
  booktitle    = {Proceedings of the Ninth Annual {ACM-SIAM} Symposium on Discrete Algorithms,
                  25-27 January 1998, San Francisco, California, {USA}},
  pages        = {473--480},
  publisher    = {{ACM/SIAM}},
  year         = {1998}
  }

@article{Vizing68,
  title={Some unsolved problems in graph theory},
  author={Vizing, Vadim G},
  journal={Russian Mathematical Surveys},
  volume={23},
  number={6},
  pages={125},
  year={1968},
  publisher={IOP Publishing}
}

@inproceedings{Moser09,
  author       = {Robin A. Moser},
  editor       = {Michael Mitzenmacher},
  title        = {A constructive proof of the Lov{\'{a}}sz local lemma},
  booktitle    = {Proceedings of the 41st Annual {ACM} Symposium on Theory of Computing,
                  {STOC} 2009, Bethesda, MD, USA, May 31 - June 2, 2009},
  pages        = {343--350},
  publisher    = {{ACM}},
  year         = {2009}
  }

@article{MoserT10,
  author       = {Robin A. Moser and
                  G{\'{a}}bor Tardos},
  title        = {A constructive proof of the general lov{\'{a}}sz local lemma},
  journal      = {J. {ACM}},
  volume       = {57},
  number       = {2},
  pages        = {11:1--11:15},
  year         = {2010}
  }

@misc{Tao09,
  author       = {Tao, Terence},
  title        = {Moser's Entropy Compression Argument},
  howpublished = {\url{https://terrytao.wordpress.com/2009/08/05/mosers-entropy-compression-argument/}},
  note         = {Blog post on \emph{What's New}, August 5},
  year         = {2009}
}

@article{Reed99c,
  title={The list colouring constants},
  author={Reed, Bruce},
  journal={Journal of Graph Theory},
  volume={31},
  number={2},
  pages={149--153},
  year={1999},
  publisher={John Wiley \& Sons, Inc. New York, NY, USA}
}

@article{ReedS02,
  title={Asymptotically the list colouring constants are 1},
  author={Reed, Bruce and Sudakov, Benny},
  journal={Journal of Combinatorial Theory, Series B},
  volume={86},
  number={1},
  pages={27--37},
  year={2002},
  publisher={Elsevier}
}

@inproceedings{BernsteinBGNSS022,
  author       = {Aaron Bernstein and
                  Jan van den Brand and
                  Maximilian Probst Gutenberg and
                  Danupon Nanongkai and
                  Thatchaphol Saranurak and
                  Aaron Sidford and
                  He Sun},
  editor       = {Mikolaj Bojanczyk and
                  Emanuela Merelli and
                  David P. Woodruff},
  title        = {Fully-Dynamic Graph Sparsifiers Against an Adaptive Adversary},
  booktitle    = {49th International Colloquium on Automata, Languages, and Programming,
                  {ICALP} 2022, Paris, France, July 4-8, 2022},
  series       = {LIPIcs},
  pages        = {20:1--20:20},
  publisher    = {Schloss Dagstuhl - Leibniz-Zentrum f{\"{u}}r Informatik},
  year         = {2022}
  }

@inproceedings{BhattacharyaSS22,
  author       = {Sayan Bhattacharya and
                  Thatchaphol Saranurak and
                  Pattara Sukprasert},
  editor       = {Shiri Chechik and
                  Gonzalo Navarro and
                  Eva Rotenberg and
                  Grzegorz Herman},
  title        = {Simple Dynamic Spanners with Near-Optimal Recourse Against an Adaptive
                  Adversary},
  booktitle    = {30th Annual European Symposium on Algorithms, {ESA} 2022, Berlin/Potsdam,
                  Germany, September 5-9, 2022},
  series       = {LIPIcs},
  pages        = {17:1--17:19},
  publisher    = {Schloss Dagstuhl - Leibniz-Zentrum f{\"{u}}r Informatik},
  year         = {2022}
  }

@article{HaeuplerP26,
  title={Maintaining Random Assignments under Adversarial Dynamics},
  author={Haeupler, Bernhard and Paramonov, Anton},
  journal={arXiv preprint arXiv:2604.05606},
  year={2026}
}

@article{HaeuplerMRST26,
  title={Dynamic Construction of the Lovasz Local Lemma},
  author={Haeupler, Bernhard and Slobodan Mitrovic and Srikkanth Ramachandran and Wen-Horng Sheu and Robert Tarjan},
  journal={Manuscript, April 2026.},
  year={2026}
}
	

\end{document}